\documentclass[10pt,conference,letterpaper]{IEEEtran}
\usepackage{times,amsmath,epsfig}
\usepackage{booktabs} 
\usepackage{booktabs} 
\usepackage{tikz}
\usepackage{amssymb}
\usepackage{algpseudocode,algorithm,algorithmicx,subfigure}
\usepackage{pdfcomment}
\usetikzlibrary{arrows}
\makeatletter
\pgfdeclareshape{datastore}{
  \inheritsavedanchors[from=rectangle]
  \inheritanchorborder[from=rectangle]
  \inheritanchor[from=rectangle]{center}
  \inheritanchor[from=rectangle]{base}
  \inheritanchor[from=rectangle]{north}
  \inheritanchor[from=rectangle]{north east}
  \inheritanchor[from=rectangle]{east}
  \inheritanchor[from=rectangle]{south east}
  \inheritanchor[from=rectangle]{south}
  \inheritanchor[from=rectangle]{south west}
  \inheritanchor[from=rectangle]{west}
  \inheritanchor[from=rectangle]{north west}
  \backgroundpath{
    \southwest \pgf@xa=\pgf@x \pgf@ya=\pgf@y
    \northeast \pgf@xb=\pgf@x \pgf@yb=\pgf@y
    \pgfpathmoveto{\pgfpoint{\pgf@xa}{\pgf@ya}}
    \pgfpathlineto{\pgfpoint{\pgf@xb}{\pgf@ya}}
    \pgfpathmoveto{\pgfpoint{\pgf@xa}{\pgf@yb}}
    \pgfpathlineto{\pgfpoint{\pgf@xb}{\pgf@yb}}
 }
}
\DeclareMathOperator*{\argmin}{\arg\!\min}
\newcommand*\Let[2]{\State #1 $\gets$ #2}

\newtheorem{example}{Example}
\newtheorem{theorem}{Theorem}
\renewcommand{\baselinestretch}{0.99}

\makeatother

\title{Eliciting Worker Preference for Task Completion}

\author{%
  {Mohammad Esfandiari{\small $~^{\#1}$}, Senjuti Basu Roy{\small $~^{\#1}$}, Sihem Amer-Yahia{\small $~^{*2}$}} 
    \vspace{1mm}\\
    \fontsize{10}{10}\selectfont\rmfamily\itshape
    $^{\#}$\,NJIT, USA
    \fontsize{9}{9}\selectfont\ttfamily\upshape
    $^{1}$\,me76@njit.edu, senjutib@njit.edu
    \vspace{1mm}\\
    \fontsize{10}{10}\selectfont\rmfamily\itshape
    $^{*}$\,Univ. Grenoble Alpes, CNRS, LIG, France\\
    \fontsize{9}{9}\selectfont\ttfamily\upshape
    $^{2}$\,sihem.amer-yahia@univ-grenoble-alpes.fr
}

\begin{document}

\maketitle
\begin{abstract}
Current crowdsourcing platforms provide little support for worker feedback. Workers are sometimes invited to post free text describing their experience and preferences in completing tasks. They can also use forums such as Turker Nation\footnote{\url http://turkernation.com/} to exchange preferences on tasks and requesters. In fact, crowdsourcing platforms rely heavily on observing workers and inferring their preferences implicitly. In this work, we believe that {\em asking workers to indicate their preferences explicitly} improve their experience in task completion and hence, the quality of their contributions. Explicit elicitation can indeed help to build more accurate worker models for task completion that captures the evolving nature of worker preferences. We design a worker model whose accuracy is improved iteratively by requesting preferences for task factors such as required skills, task payment, and task relevance. We propose a generic framework, develop efficient solutions in realistic scenarios, and run extensive experiments that show the benefit of explicit preference elicitation over implicit ones with statistical significance. 
\end{abstract}

\section{Introduction}

The main actors of a crowdsourcing platform are tasks and workers who complete them. A range of studies point out the importance of designing incentive schemes, other than financial ones, to encourage workers during task completion~\cite{BedersonQ11,KitturNBGSZLH13}. In particular, it is expected that a crowdsourcing system should ``{\em achieve both effective task completion and worker satisfaction}''. The ability to characterize the workforce with {\em factors that influence task completion} is recognized to be of great importance in building such a system~\cite{DBLP:conf/dbcrowd/RoyLTAD13, amer2016human, kaufmann2011more, DBLP:journals/corr/abs-1210-0962, DBLP:conf/sigecom/HortonC10, MartinHOG14, RogstadiusKKSLV11,CSCW17,DBLP:conf/cscw/DaiRPC15}.
Those efforts have focused on implicitly observing workers and inferring their preferences. In this paper, we argue that solely relying on implicit observations does not suffice and propose to {\em elicit preferences from workers} explicitly, as they complete tasks. Any computational model, designed for the workers to understand their task completion likelihood needs to consume worker preference. The evolving nature of worker preference requires to periodically ask workers and refine such models. To the best of our knowledge, this work is the first to examine the benefit of explicit preference elicitation from workers and its impact on effective task completion. Our proposed approach of explicit preference elicitation does not incur additional burden to the workers; in fact, platforms such as, Amazon Mechanical Turk\footnote{\url{https://www.mturk.com/}} already seek worker feedback in the form of free text. Our effort is to judiciously select questions for elicitation in a stuctured and holistic fashion.

In this paper, {\em our objective is to design a framework that advocates for explicit preference elicitation from workers to develop a model that guides task completion}. That differs from developing solutions for task assignment. The objective behind preference elicitation is to use obtained feedback to effectively maintain a \textit{Worker Model}. Given a task $t$ that a worker $w$ undertakes (either via self-appointment or via an assignment algorithm), there could be one of two possible outcomes : 1. the task is completed successfully. 2. otherwise. In reality, worker preferences are {\em latent}, i.e., they are to be inferred through task factors and task outcomes. Popular platforms, such as Mechanical Turk or Prolific Academic,\footnote{\url{https://www.prolific.ac/}}, have characterized tasks using factors, such as {\em type, payment} and {\em duration}.
While our framework is capable to consume any available task factors, our effort nevertheless is to  propose a \textit{ generic solution that characterizes workers by understanding their preferences for a given set of task factors}~\cite{DBLP:conf/dbcrowd/RoyLTAD13,amer2016human}.

Our overarching goal is to seek feedback from workers to effectively maintain a \textit{Worker Model} for task completion. To achieve this goal, the  first challenge is to define an accurate model that predicts, per worker, how much each task factor is responsible for the successful completion of that task or for its failure. 
 We propose to {\em bootstrap} this model by selecting a small set of tasks a worker needs to complete initially to learn her model. 
An equally important challenge is to update the \textit{Worker Model}, as workers complete tasks. Indeed, unless that model is updated periodically, it is likely to become outdated, as worker's preferences evolve over time.
To update the model, we advocate the need to explicitly elicit from a worker her preferences. That is a departure from the literature where workers are observed and their preferences computed implicitly. We claim that the explicit elicitation of preferences results in a more accurate \textit{Worker Model}. Preferences are elicited via the {\bf Question Selector} that selects a set of $k$ task factors and asks a worker $w$ to rank them. For example, a worker may be asked {\em``Rank task relevance and payment''}. A higher rank for payment will indicate the worker's preference for high paying tasks over those most relevant to her profile. Once the worker provides her preference, the \textit{Worker Model} is updated with the help of the {\bf Preference Aggregator}. {\bf Question Selector} and {\bf Preference Aggregator} constitute the two computational problems of our framework.

{\bf Worker Model.} Our natural choice is to use a graphical model~\cite{koller2009probabilistic}, such as a Bayesian Network where each node is a random variable (task factors/worker preference/task outcome) and the structure of the graph expresses the conditional dependence between those variables. The observed variables are task factors and task outcomes and the \textit{Worker Model} contains worker preferences in the form of latent variables that are inferred through that model. 
It is however known that structure learning in Bayesian Network is NP-hard~\cite{koller2009probabilistic} and that the parameters could be estimated through methods such as Expectation Maximization, that also are computationally expensive. As both {\bf Question Selector} and {\bf Preference Aggregator} have to invoke the model many times,  it becomes prohibitively expensive to use it in real time. We therefore propose a simplified model that has a one-to-one correspondence between task factors and worker preference. The preference of a worker for a task factor is construed as a weight and the \textit{Worker Model} becomes a linear combination of the task factors. This simplification allows us to design efficient solutions.

{\bf Question Selector.} The question selector intends to select the $k$-task factors whose removal maximizes the improvement of the \textit{Worker Model} $\mathcal{F}$. The idea is to present those {\em uncertain} factors to the worker and seek her explicit preference. We prove that optimally selecting $k$ questions, i.e., $k$ task factors for a worker, is NP-hard, even when the \textit{Worker Model} is linear. We develop an efficient alternative using an iterative greedy algorithm that has a provable approximation bound.

{\bf Preference Aggregator.} The second technical problem is to update the \textit{Worker Model} with the elicited preference. Given a set of $k$ task factors, worker $w$ provides an absolute order on these factors. The obtained ranking is expressed as a set of $k(k-1)/2$ pairwise linear constraints, as $i > j$, $i > l$, etc. We design an algorithm that updates the \textit{Worker Model} using the same optimization function as the one used to build it initially, modified by adding those constraints. With a Bayesian Network as the underlying model, the addition of dummy variables would encode the constraints aptly. However, with one variable per constraint, the solution would not scale. For the simplified linear \textit{Worker Model}, we add them as pairwise linear constraints. The problem then becomes a constrained least squares problem that could be solved optimally in polynomial time.

We run experiments that measure the accuracy of our model and the scalability of our approach with real tasks and workers: $165,168$ tasks from CrowdFlower involving $58$ workers from Amazon Mechanical Turk. We measure the accuracy of the {\em Worker Model} against several baselines: a random selection of which task factors to invoke preferences for, and implicit preference computation~\cite{edbtMata}. We show that soliciting preferences explicitly and using them to update the model greatly reduces error and largely outperforms implicit solutions with statistical significance. We also show that our approach scales well.

In summary, our contributions are: \begin{itemize}
\item Problem Formalism (Section~\ref{sec:model}): A framework that has a \textit{Worker Model} that captures, per worker, her preference for task factors to predict the likelihood of task completion. 
We present an innovative formulation to bootstrap the model. An important aspect of the model is that it could be easily adapted to other crowdsourcing processes, such as, task assignment or worker compensation.
We present two core problems around the model: {\bf Question Selector} that asks a worker to rank the $k$ task factors that cause the highest error in the model, and {\bf Preference Aggregator} that updates the model with elicited preferences.
\item Technical Results (Section~\ref{sec:solutions}): We study the hardness of our problems and their reformulation under realistic assumptions, as well as design efficient solutions with provable guarantees.
\item Experimental Results (Section~\ref{sec:experiments}): We present extensive experiments that corroborate that explicit preference elicitation outperforms implicit preferences~\cite{edbtMata}, and that our framework scales well.
\end{itemize}


\section{Formalism and Framework}
\label{sec:model}

We present our formalism, following which we provide an overview of the proposed framework and the problems we tackle.

\begin{example}\label{ex1}
  We are given a set of tasks, where each task is characterized by a set of factors (e.g.,  {\em type, payoff, duration} are some examples). A task could be of different types, such as, image tagging, ranking, sentiment analysis. Payoff determines the \$ value the workers receives as payment, whereas, duration is an indication of the time a worker needs to complete that task. Generalizing this, one can imagine that each task could be described as a vector of different factors and a set of tasks together gives rise to a task factor matrix $\mathcal{T}^f$. One such matrix of $6$ tasks is presented below:
\begin{equation*}
  \begin{footnotesize}
    \begin{Bmatrix}
    \label{task_factor}
        {\mathit{id}} & {\mathit{tagging}} & {\mathit{ranking}} & {\mathit{sentiment}} & {\mathit{payoff}} & {\mathit{duration}} & {\mathit{outcome}} \\
        t_1 & 1 & 0 & 0 & high & long & 1 \\
        t_2 & 1 & 0 & 0 & low & short & 0 \\
        t_3 & 0 & 1 & 0 & low & short & 1 \\
        t_4 & 0 & 1 & 0 & low & long  & 0\\
        t_5 & 0 & 0 & 1 & high & short & 1 \\
        t_6 & 0 & 0 & 1 & high & long & 0 \\
    \end{Bmatrix}
    \end{footnotesize}
\end{equation*}
Given a task $t$ that a worker $w$ undertakes (either via self-appointment or via an assignment algorithm), there could be one of two possible outcomes : 1. the task is completed successfully (denoted by $1$). 2. otherwise (denoted by $0$). The last column of the task factor matrix indicates outcomes of each of the tasks. These could be known from prior history or predicted using a mathematical model.
\end{example}

\subsection{Formalism}
{\bf Task Factors.} In a crowdsourcing platform, each task $t$ in a set of $n$ given tasks $\mathcal{T}=\{t_1,...,t_n\}$ is characterized by a set of $m$ factors whose values are either explicitly present or could be extracted. For this work, we assume that for a task $t$, its factors give rise to a vector $\vec{t}^f$ that is given and we do not focus on how to obtain them. 
 This gives rise to the task-factor matrix $\mathcal{T}^f$ of dimension $n \times m$. For the simplicity of exposition, task factors are presented as binary, although our proposed solutions adapt when they are continuous or categorical.

Using Example~\ref{ex1}, the task factor matrix consists of $6$ tasks and each task is described by $5$ factors.

{\bf Worker Preferences.} The preferences of a worker $w$ are represented by a vector $\vec{w}^f$ of length $p$ that takes real values and determines the preferences of $w$ for tasks. Using Example~\ref{ex1}, $\vec{w}^f$ could be represented as {\em latent variables}, such as, $\{\textit{skill, motivation, reputation}\}$. The correspondence between task factors and worker preference is surjective (e.g., task type $\rightarrow$ skill, $\{duration, \text{payoff}\}$ $\rightarrow$ motivation). Worker preference variables cannot be observed and must be inferred.

{\bf Explicit Questions.} An explicit question is asked to elicit $w$'s preference on a particular task factor, as the tasks dictate worker preference and hence her performance thereof. In fact, there is an one to one correspondence between the questions and task factors. A set of $k$ questions is asked to obtain a preferred order among a set of $k$ task factors (where $k$ is part of the input).  As an example, one may ask to rank {\em ``task duration, tagging tasks, ranking tasks, sentiment analysis tasks, payment''}. A worker provides an absolute order among these $5$ factors as her preference. The ranking can be simply interpreted as the worker's preference for the first factor followed by the second, etc. For example, if the worker ranks {\em payment} first, she means a preference for high paying tasks to any others including those most relevant to her profile.


\subsection{Framework and Challenges}\label{frame}
We propose an iterative framework (refer to Figure~\ref{system_workflow}) that is designed to ask personalized questions to a worker to elicit her preference in a crowdsourcing platform. The rationale for our proposal is that while task factors are stable, a worker's preference evolves as workers undertake tasks. Workers' skills improve as they complete tasks and their motivation varies during task completion~\cite{CSCW17,edbtMata}. How to define an iteration is orthogonal to our problem. Indeed, an iteration could be
defined by discretizing time into equal-sized windows, by the number of available/completed tasks, by the number of workers, or a combination thereof. We will see in Section~\ref{sec:experiments} how we define an iteration in our experiments.


Central to our framework is a model  $\mathcal{F}$ that consumes task factors and given a worker's history, infers her preference vector to predict the outcome of a new task to be undertaken by her. 
Even though we develop a worker model to predict task completion quality, this information could be used in many places to characterize the workforce of a crowdsourcing platform and enable several improvements such as the analysis of workers' fatigue~\cite{CSCW17} and motivation, as well as task assignment~\cite{ZhengWLCF15, HoV12, HoJV13, DBLP:journals/vldb/RoyLTAD15, edbtMata}. However, unless the \textit{Worker Model} is refreshed or updated periodically, it is likely to become outdated, as worker preference evolves over time~\cite{CSCW17,amer2016human,edbtMata}. To update the model, one has to periodically invoke an explicit preference elicitation step through {\bf Question Selector} that selects a set of $k$ task factors and asks worker $w$ to rank them. Once the worker provides her preference, the \textit{Worker Model} is updated by the {\bf Preference Aggregator}. These two components form the heart of our computational challenges. {\em The first is to quickly select the bst set of $k$ factos to invoke feedback for. The second is to quickly relearn the model while satisfying the answers the worker has provided.}

\begin{table}
\begin{tabular}{|c|l|}
  \hline
  Notation & Definition \\
  $\mathcal{T}$ & a set of tasks $\{t_1,\ldots,t_n\}$ \\
  $\vec{t}^f$         & a vector describing task factors \\
  $\mathcal{T}^f$ & task factor matrix \\
  $\vec{w}^f$           & worker $w$'s preference vector \\
  $\mathcal{Q}$ & a set of questions (task factors)\\
  $\mathcal{Q}^s$ & a set of selected $k$ questions\\
  $t^{O}$     & a label signaling task outcome\\
  $\mathcal{T}^{O}$     & labels signaling task outcomes of a set $\mathcal{T}$\\
  $\mathcal{B}$ & a set of $b$ tasks for bootstrapping \\
  $\mathcal{E}$ & reconstruction error of the model $\mathcal{F}$ \\
  \hline
\end{tabular}
\caption{Table of important notations}
\end{table}

\begin{figure}[t]
  \centering
      \includegraphics[width=0.5\textwidth]{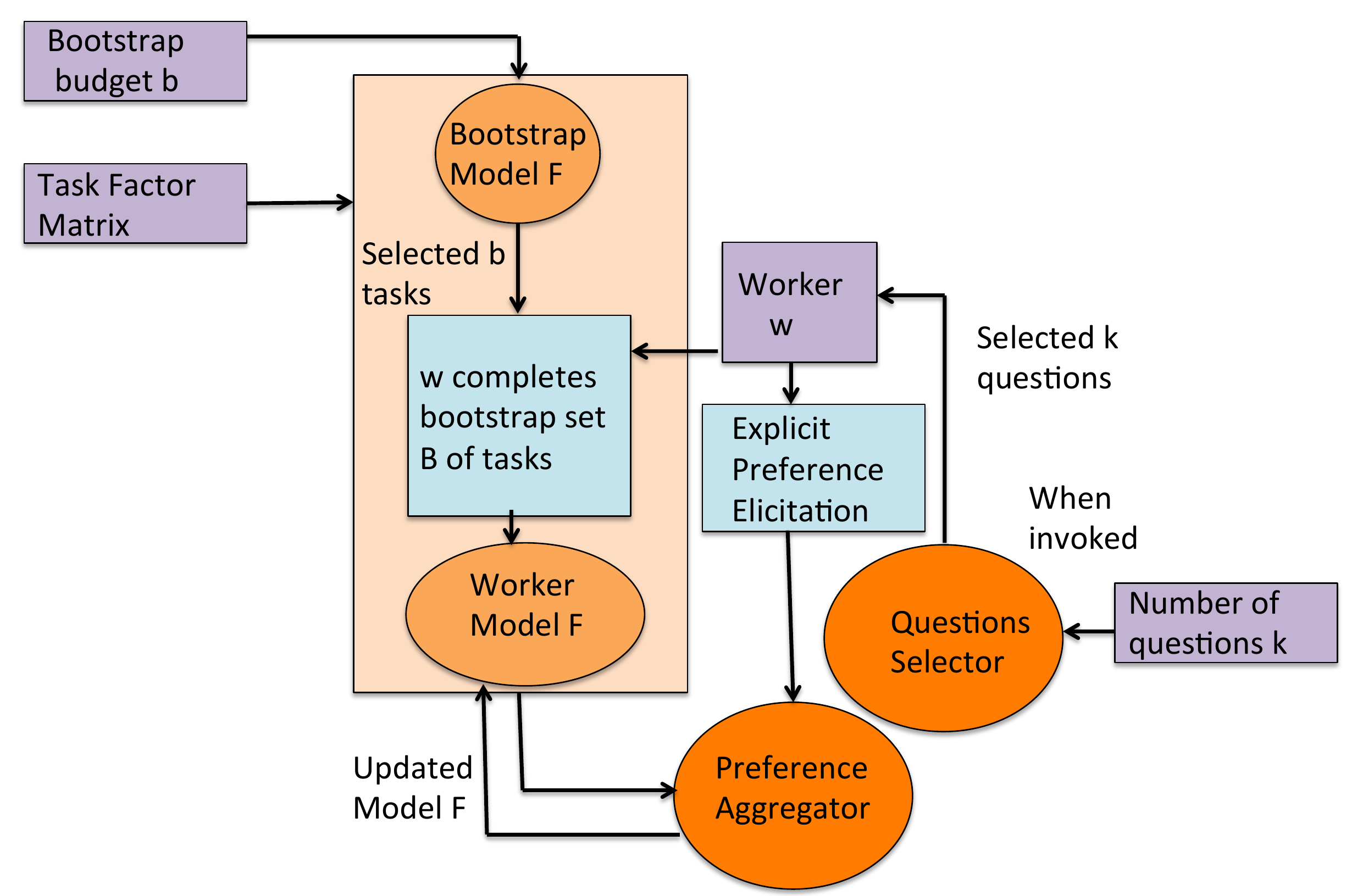}
  \caption{Explicit Preference Elicitation Framework}\label{system_workflow}
\end{figure}

The remainder of the paper focuses on a particular worker $w$, unless otherwise stated - i.e., each of the components of the framework is designed or invoked for her.

\vspace{-0.1in}
\subsection{Problem Definitions}\label{prob}

{\bf Worker Model.} Given the task factor matrix $\mathcal{T}^f$ of a set of tasks $\mathcal{T}$, where each task $t$ is associated with a known outcome ($t^{O}=1/0$ successfully completed or not), the \textit{Worker  Model}  $\mathcal{F}$ estimates the preference vector $\vec{w}^f$ of a worker $w$. $\mathcal{F}$ is a function of $\mathcal{T}^f$ and $\vec{w}^f$ , denoted as $\mathcal{T}^f \otimes \vec{w}^f$. The correctness of $\mathcal{F}$ is estimated using \textit{its reconstruction error}, i.e., $\mathcal{E}(\mathcal{T}^{O} -  (\mathcal{T}^f \otimes \vec{w}^f))$ is minimized.

Once the model is built, it can predict the outcome of a  task undertaken by $w$. On Example~\ref{ex1}, the  model predicts the likely outcome of each of the $6$ tasks, by consuming $\mathcal{T}^f$ and inferring $\vec{w}^f$.

{\bf Bootstrapping.} Initially, as no past history of a worker $w$ is available, she is treated akin to a ``cold user'' and the {\em bootstrap process} selects a subset of tasks, $\mathcal{B} \subset \mathcal{T}$ to be completed by $w$. The size of  $\mathcal{B}$, denoted by $b$, is given as a budget to bootstrapping.  The idea is to leverage the contributions of $w$ for those $b$ tasks to estimate her preference and build the \textit{Worker Model} $\mathcal{F}$. During bootstrapping, the algorithm is {\em offline}, i.e., it apriori decides all $b$ tasks, without actually observing outcomes of tasks completed by $w$. For that, it selects the $b$ tasks whose feedback minimizes the {\em expected reconstruction error} over the remaining $\mathcal{T} - \mathcal{B}$ tasks, i.e., $\mathop{\mathbb{E}}((\mathcal{T} - \mathcal{B})^{O} -  ((\mathcal{T} - \mathcal{B})^f \otimes \vec{w}^f))$ is minimized. The same approach is adopted in subsequent steps to refine $\mathcal{F}$.

Given Example~\ref{ex1} with $b=3$, the objective would be to select $3$ tasks that  minimize the expected reconstruction error.

We are now ready to describe the two fundamental problems that our system needs to solve.

{\bf Question Selector.} This module selects the best set of $k$ questions for a worker $w$. The objective is to select those task factors that are responsible for the model's inaccuracy, i.e., removing them would improve the most the reconstruction error of $\mathcal{F}$. Let $\mathcal{E}$ denote the current reconstruction error of $\mathcal{F}$ and $\hat{\mathcal{E}}$ denote it when $k$ task factors are removed. Given $\mathcal{Q}$, the $k$ questions are selected such that the model reconstruction error improves the most, i.e., $argmax_{\{\mathcal{Q}^s \in \mathcal{Q}:|\mathcal{Q}^s|=k\}} (\mathcal{E} - \hat{\mathcal{E}}_{m-\mathcal{Q}^s})$.

Using Example~\ref{ex1}, if $k=2$, this will select any two of the five task factors in the task factor matrix.

{\bf Preference Aggregator.} The preferences $\mathcal{P}$ provided by a worker for task factors, could be expressed as {\em a set of $k(k-1)/2$ pairwise linear constraints} of the form, $i>j$, $i >l$, $j>l$. Worker's preferences are taken as hard constraints. Given $\mathcal{P}$, the objective is to relearn $\mathcal{F}$ that satisfies $\mathcal{P}$ such that its reconstruction error is minimized. The objective therefore is to minimize $\mathcal{E}(\mathcal{T}^{O} -  (\mathcal{T}^f \otimes \vec{w}^f))$ such that the constraints in $\mathcal{P}$ are satisfied.

Using Example~\ref{ex1}, if worker $w$ explicitly states that she prefers {\tt annotation} tasks to {\tt ranking} tasks, this preference is translated into constraints expressed on the worker preference vector. Those are then used by the preference aggregator to update $\mathcal{F}$.

\section{Solutions}
\label{sec:solutions}

We now propose a generic approach for building and bootstrapping the \textit{Worker Model} and for solving the two computational problems enunciated in Section~\ref{sec:model}, namely, {\bf Question Selector} and {\bf Preference Aggregator.} In each case, we present our generic approach and then develop a simplified, yet realistic framework, that enables efficient solutions with theoretical guarantees.

\subsection{Worker Model}\label{genmodel}
As described in Section~\ref{frame}, central to our framework is a supervised model $\mathcal{F}$, designed for a worker $w$, that consumes task factors and predicts task outcomes by inferring $w$'s preferences for those factors.
There are several machine learning models that could potentially be used.  A natural choice is a probabilistic graphic model~\cite{koller2009probabilistic}, such as, Bayesian Networks, where each node is a random variable (task factors/worker preferences/task outcome) and the structure of the graph expresses the conditional dependence between them. In particular, it could be expressed as a Bayesian Network, which is represented as a directed acyclic graph (DAG). As described in Figure~\ref{pgm}, the observed variables are task factors and task outcomes and the worker preferences are captured as latent variables that are inferred.

Furthermore, we can impose constraints, such as, the task factors are independent from each other, and the worker preference variables are correlated/not, or the worker preference variables determine task outcomes, which are in turn dependent on the task factors. Given a set of tasks, each with an outcome and associated vector of factors, the model $\mathcal{F}$ corresponds to the factorization of the joint probability distribution over these variables:
\begin{equation}
\begin{split}
Pr(\text {task factors, worker preferences, outcome}) =  \\
 Pr(outcome|\text {worker preferences}) \times \\
\prod_{i=1}^{p} Pr(\text {worker preference i}| \text {parent task factors of i}) \\
\times \prod_{j=1}^{m} Pr(\text {task factor j})
\end{split}
\end{equation}

There are two primary computational difficulties when building $\mathcal{F}$ : learning the structure of the graph and estimating parameters to obtain a probability distribution at each node. The overall objective is to minimize reconstruction error. Once the model is built, we use it for inference, i.e., it will infer the probability distribution over a worker's preference, given the values of task factors and task outcomes.

The proposed non-linear model is prohibitively expensive to compute. The structure learning part is known to be NP-hard~\cite{koller2009probabilistic} while the parameters could be estimated through methods such as Expectation Maximization, also computationally expensive. The known efficient algorithms that propose a workaround are primarily heuristics~\cite{margaritis2000bayesian} without approximation guarantees.

Therefore, we suggest to {\em simplify our model under the assumption that the preference of a worker for a task factor is a weight and the \textit{Worker Model} is a linear combination of a given set of task factors}. As a result, the \textit{Worker Model} can be expressed as a linear regression function. A one-to-one correspondence between task factors and worker preferences (\# task factors = length of $\vec{w}^f$, i.e., $p=m$) would  treat  worker preferences  as weights that are to be estimated. The model takes the form,
\begin{equation*}
\mathcal{T}^{O} = \vec{w}^f \times \mathcal{T}^f
\end{equation*}
Given the task factor matrix $\mathcal{T}^f$  and the observed task outcome vector $\mathcal{T}^{O}$, the objective is to estimate $\vec{w}^f$, such that the reconstruction error is minimized, i.e.,

\begin{equation}\label{ls}
\argmin_{\vec{w}^f \in \mathbb{R}^m} || \mathcal{T}^{O} - \vec{w}^f \times \mathcal{T}^f||_2
\end{equation}

\begin{figure}
  \centering
      \includegraphics[width=0.4\textwidth]{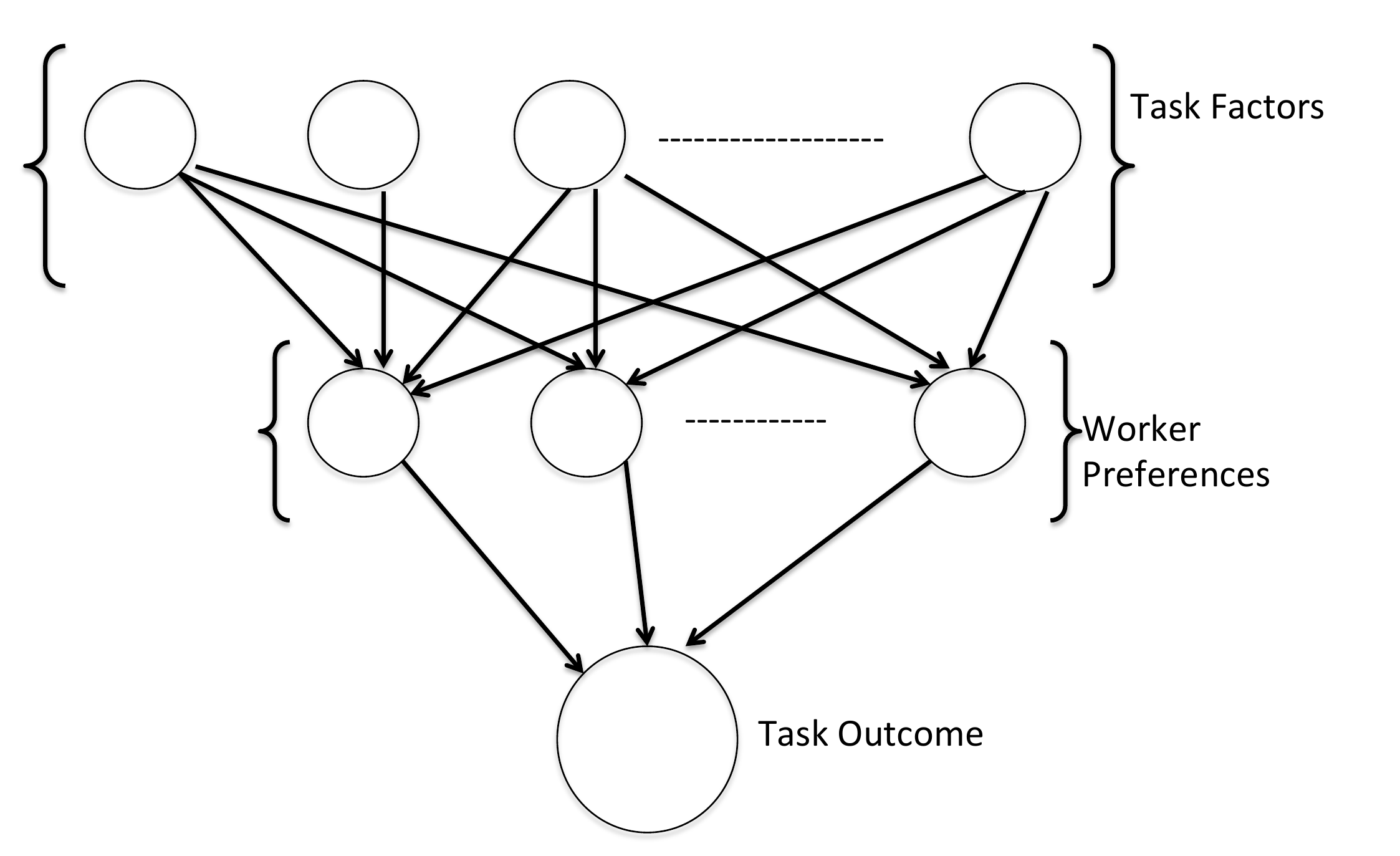}
  \caption{Worker Model}\label{pgm}
\vspace{-0.1in}
\end{figure}

{\bf Algorithm for Worker Model.}
As long as the task factor matrix $\mathcal{T}^f$ is invertible, or could be inverted by adding an additional term~\cite{DBLP:journals/corr/BiswasLR17}, the objective function expressed in Equation~\ref{ls} has an equivalent and alternative representation of the form ${({\mathcal{T}^f}^T\mathcal{T}^f)}^{-1}{\mathcal{T}^f}^T\mathcal{T}^{O}$, where ${({\mathcal{T}^f}^T\mathcal{T}^f)}^{-1}{\mathcal{T}^f}^T$ is known as the Moore-Penrose pseudo-inverse matrix of $\mathcal{T}^f$. The proof could be found in~\cite{albert1972regression}.

We design an ordinary least squares (OLS)~\cite{dismuke2006ordinary} solution to  estimate the regression coefficient or worker preferences $\vec{w}^f$. It first transforms the objective function in Equation~\ref{ls} to a Moore-Penrose pseudo-inverse matrix of $\mathcal{T}^f$. This equivalent representation is proved to have a global minimum, hence the obtained OLS estimator is optimal. The overall running time of OLS is dictated by matrix multiplication (multiplying ${({\mathcal{T}^f}^T\mathcal{T}^f)}$,  Matrix inversion (inverting ${{\mathcal{T}^f}^T\mathcal{T}^f)}^{-1}$, followed by matrix multiplication (to obtain ${({\mathcal{T}^f}^T\mathcal{T}^f)}^{-1}{\mathcal{T}^f}^T$), and a final matrix multiplication (to obtain ${({\mathcal{T}^f}^T\mathcal{T}^f)}^{-1}{\mathcal{T}^f}^T$). The asymptotic complexity is $\mathcal{O}(m^2n+m^3)$.

\textbf{Bootstrapping the Worker Model.}\label{bootspec}
Unfortunately, the \textit{Worker Model} can not be built for a brand new worker, who has not completed any task before in the platform. For such workers, we propose to bootstrap the \textit{Worker Model}. Bootstrapping is a common practice in recommender systems and is adopted by largely used platforms such as Netflix and MovieLens.\footnote{\url{https://www.netflix.com/, https://movielens.org/}}. The objective in our case is to select the set of tasks $\mathcal{B}$ and learn the \textit{Worker Model} $\mathcal{F}$ that minimizes the {\em expected reconstruction error} over the remaining $\mathcal{T} - \mathcal{B}$ tasks. For a selected set of $b$ tasks, we compute a set of $2^b$ \textit{Worker Models}, where each model is learned by encoding one of the $2^b$ possible combinations of the task outcomes, and capturing the reconstruction error over the remaining set of $\mathcal{T} - \mathcal{B}$ based on the learned model. This gives rise to a bootstrapping tree with $2^b$ branches, as the one shown in Figure~\ref{bstp}. Each branch is associated with a probability value that represents the probability of that combination of $b$ task outcomes. As an example, in Figure~\ref{bstp}, the leftmost branch captures the model where all three tasks would have a successful completion, the branch represents that probability, and the corresponding leaf represents the reconstruction error.

{\bf Generic and Simplified Probability Model for Bootstrapping.} To capture the probability of a branch in the bootstrap tree, we need to calculate the probability of successful/unsuccessful outcome of each task by the ``cold'' worker. Technically, we are interested to compute $Pr(t^0 = 1|\vec{t}^f,w)$ (probability of successful completion) and $Pr(t^0 = 0|\vec{t}^f,w)$ (probability of unsuccessful completion). In a real crowdsourcing platform, however, there is little to no information available about a new worker that could be potentially used to capture similarity between her and other existing workers in the system. Thus, the only way to obtain this information is to see if there are {\em past tasks with similar characteristics} (although undertaken by different workers) and analyze the outcome of those tasks. Therefore, these values should rather rely only on task factors, $Pr(t^0 = 1|\vec{t}^f)$ and $Pr(t^0 = 0|\vec{t}^f)$  should be computed as the joint distribution of the task factors and outcomes. Using Example~\ref{ex1}, if we assume that the task factors, such as, duration, payoff, and the task types are correlated with each other, then the outcome of a task relies on the joint distribution over these factors.

The optimal solution of this problem could be obtained by computing the {\em joint distribution} using a structure similar to a contingency table where each cell represents a possible set of values for each of the $m$ task factors and the value of that cell represents the probability of successful completion. Classical algorithms such as iterative proportional fitting (or IPF)~\cite{fienberg1970iterative}, could be used for estimating this joint distribution. Unfortunately, these algorithms do not scale when the number of factors and their possible values are large~\cite{fienberg1983iterative}.

Given a task $t$, $Pr(t^0 = 1|\vec{t}^f)$ (probability of successful completion) and $Pr(t^0 = 0|\vec{t}^f)$ (probability of unsuccessful completion), s the joint distribution over the task factors and outcomes is very expensive, as this will require us to compute a joint distribution over a $v^m$ space, if each of the $m$ task factor variables takes $v$ possible values. A more realistic probability model relies on a conditional independence assumption. It assumes that the task factors are themselves independent but the task outcome is conditionally dependent on each of the task factors. This Bayesian assumption is not an overstretch. Using Example~\ref{ex1}, one can see that the different task types, annotation, ranking, or sentiment analysis are independent from each other, as is duration, but the task outcome is conditionally dependent on each of these factors. In practice, some dependence may exist between factors, e.g., task duration and payment may be correlated. However, that is not the case in general as tasks are posted by different requesters and in different countries. Therefore, we can confidently claim that a conditionally independent probability model captures a large number of cases in practice.

Under the conditional independence assumption, we have
\begin{equation}\label{specprob}
Pr(t^0 = 1|\vec{t}^f) = \prod_{i \in \vec(t)^f} Pr(t^0 = 1 | t^{(i)})
\end{equation}
Using Bayes' Theorem, this could be rewritten as
\begin{equation*}
Pr(t^0 = 1|\vec{t}^f)  = \prod_{i \in \vec(t)^f} \frac{Pr((t^0 = 1) | t^{(i)}) \times Pr(t^0 = 1) }{Pr(t^{(i)})}
\end{equation*}
Computing  the  probability  formula requires us to know the value of quantities such as
$Pr((t^0 = 1) | t^{(i)})$, $Pr(t^0 = 1)$, and $Pr(t^{(i)})$. However, singleton
and pairwise probabilities can be computed in a pre-processing
step  considering other tasks and workers.  For  example, $Pr((t^0 = 1) | t^{(i)})$ can  be  estimated  as  the  fraction  of  previous tasks with successful outcomes
that also have the $i$-th factor as $t^{(i)}$ and $Pr(t^0 = 1)$ can be estimated as the fraction of tasks with successful outcomes, whereas, $Pr(t^{(i)})$ is the fraction of tasks that have $t^{(i)}$ as the $i$-th factor. Each of these calculations are efficient and could be done in a pre-processing phase.

Assuming independence among tasks, the probability of a combination of outcomes of $b$ tasks is then obtained by multiplying the probabilities of the individual task outcomes. Considering Figure~\ref{bstp}, the probability of the left most branch is $Pr({t_a}^0 = 1|\vec{t}^f) \times Pr({t_b}^0 = 1|\vec{t}^f) \times Pr({t_d}^0 = 1|\vec{t}^f)$

{\bf Bootstrapping Algorithm.} Given a budget of $b$ tasks, these tasks are chosen a priori, therefore, one has to explore all possible $\binom{n}{b}$ tasks, and for each task, compute its two possible outcomes probabilistically. This gives rise to an exponential search space of $\binom{n}{b}$ possible choices. Even when the best set of $b$ tasks are chosen, each task has one of two possible outcomes, which gives rise to a bootstrap tree with $2^{b}$ branches~\cite{mottin2013probabilistic}, as shown in the Figure~\ref{bstp}.  Each branch from the root to the leaf corresponds to a set of $b$ tasks and their outcomes, and each leaf is associated with a reconstruction error. The reconstruction error of the branch is the product of the respective probabilities of the outcomes and the reconstruction error of the model $\mathcal{F}$ thereof. The expected reconstruction error is the sum of the errors over all $2^{b}$ branches. The objective, as mentioned in Section~\ref{prob}, is to design the tree that improves the expected reconstruction error the most.

\begin{figure}
  \centering
      \includegraphics[width=0.4\textwidth]{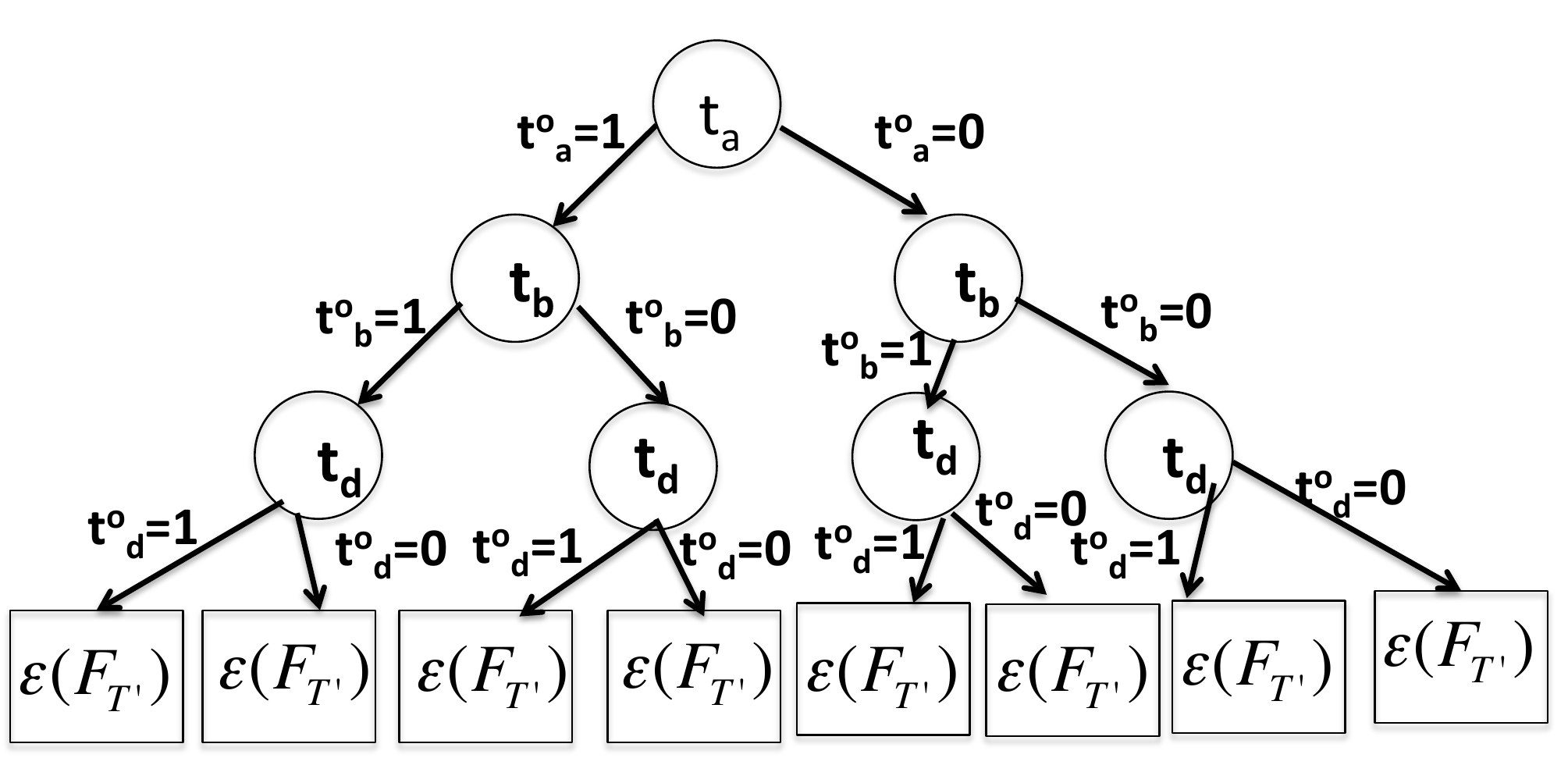}
  \caption{Bootstrapped tree of three chosen tasks $\{t_a,t_b,t_d\}$}\label{bstp}
\vspace{-0.1in}
\end{figure}

\begin{theorem}\label{th1}
Bootstrapping the \textit{Worker Model} is NP-hard.
\end{theorem}

\begin{proof}
(Sketch:) If the underlying model is a graph whose structure computation is NP-hard~\cite{koller2009probabilistic}, naturally the bootstrapping becomes NP-hard, as it has to invoke that model as a subroutine. Even for an arbitrarily simple model that is polynomial time computable, the NP-hardness could be proved using a reduction from the Set Cover Problem, even when each factor is only binary~\cite{garey1972optimal}.
\end{proof}

We now design an algorithm that is greedy in nature and avoids searching the $\binom{n}{b}$ space to select the $b$ tasks. It runs in $b$ iterations and in the $i$-iteration selects the task out of the remaining set of tasks that has the highest marginal improvement over the objective function. However, even when the $b$ tasks are selected, computing the bootstrap tree is exponential in $b$ (recall Figure~\ref{bstp}). This algorithm makes $\mathcal{O}(n \times b)$ comparisons to select the best set of $b$ tasks. However, while the greedy selection is in progress, given an already selected $b'$ tasks, it still has to build the bootstrap tree with  $2^{b'}$ branches. The worst case asymptotic complexity is  therefore, $\mathcal{O}(n \times b \times 2^{b})$.

{\bf Running Example:} Using Example~\ref{ex1}, tasks $\{t_1,t_3,t_6\}$ are chosen for bootstrapping. Once $\mathcal{F}$ is developed, it assigns the following weights to the task factors. tagging=$0.4$, ranking= $0.69$, sentiment=$0.1$, payoff=$0.4$, duration           =$0.42$. This is intuitively explainable, as the tasks that the worker complete successfully are tagging and ranking (hence gets higher weights), but the sentiment analysis task is not completed successfully (thus, gets a lower preference value).

\subsection{Question Selector}\label{genq}
The objective is to select $k$ questions (task factors) eliminating which would maximize the reconstruction error reduction of the model $\mathcal{F}$.
Ideally, out of $m$ task factors (a set $\mathcal{Q}$ of questions),  $k$ factors should be chosen as a set. This gives rise to an exponential search space that could be modeled using a decision tree like structure with~\cite{mottin2013probabilistic}, $\binom{m}{k}$ possible branches. Each branch from the root to the leaf corresponds to a set of $k$ questions, and the leaf is associated with a reconstruction error. The objective, as mentioned in Section~\ref{prob}, is to design that tree that improves the reconstruction error the most.

\begin{theorem}
Optimally selecting $k$ questions to elicit preferences is NP-hard.
\end{theorem}

\begin{proof}
(Sketch:) A careful review of the objective function (refer to Section~\ref{prob}) shows that since $\mathcal{E}$ is a constant at a given point - thus, maximizing $(\mathcal{E} - \hat{\mathcal{E}}_{m - \mathcal{Q}^s}):\{\mathcal{Q}^s \in \mathcal{Q}:|\mathcal{Q}^s|=k\} $ is same as minimizing the reconstruction error of  $\hat{\mathcal{E}}_{m - \mathcal{Q}^s}$, i.e., retaining the best $m-k$ factors (thus eliminating the worst $k$ factors) that has the smallest reconstruction error of $\mathcal{F}$.
The problem thus becomes selecting the best $m-k$ factors that have the smallest reconstruction error. The remaining $k$ factors would therefore be chosen as the explicit questions for preference elicitation.

The reduction is done using the Set Cover problem. We omit the details for brevity, but elaborate later on that the problem remains NP-hard even when we consider a simple \textit{Worker Model}.
\end{proof}

{\bf Efficient Solution for Question Selector.}
Since $\mathcal{E}$ is a constant at a given point, maximizing $(\mathcal{E} - \hat{\mathcal{E}}_{m - \mathcal{Q}^s}):\{\mathcal{Q}^s \in \mathcal{Q}:|\mathcal{Q}^s|=k\} $ is same as minimizing the reconstruction error of  $\hat{\mathcal{E}}_{m - \mathcal{Q}^s}$, i.e., retaining the best $m-k$ factors (thus eliminating the worst $k$ factors) that has the smallest reconstruction error of $\mathcal{F}$. The problem thus becomes selecting the best $m-k$ factors that have the smallest reconstruction error. The remaining $k$ factors would therefore be chosen as the explicit questions for preference elicitation.

\begin{theorem}
Optimally selecting $k$ questions for explicit worker preference is NP-hard even when the \textit{Worker Model} is linear.
\end{theorem}

\begin{proof}
(Sketch): When a linear model such as the one in Equation~\ref{ls} is assumed, as described above, the objective function is equivalent to minimizing ${({\mathcal{T}^f}^T\mathcal{T}^f)}^{-1}{\mathcal{T}^f}^T\mathcal{T}^{O}$, where ${({\mathcal{T}^f}^T\mathcal{T}^f)}^{-1}{\mathcal{T}^f}^T$ is known as the Moore-Penrose pseudo-inverse matrix of $\mathcal{T}^f$.

Therefore, the problem of identifying and removing the $k$ worst factors, i.e., retaining the best $m-k$ factors, is akin to selecting a subset of $m-k$ columns from the task factor matrix $\mathcal{T}^f$ such that the pseudo-inverse of this sub-matrix has the smallest norm. Under the Frobenius or $L_2$ norms this problem is proved to be NP-hard~\cite{DBLP:journals/corr/BiswasLR17}.

Using the NP-hardness proof described in~\cite{DBLP:journals/corr/BiswasLR17}, our reduction is rather simple. Given an instance of that problem, we set $k$ (the $k$ worst factors to remove) as the difference between the total number of columns and $k'$ ($k'$= the best set of $k'$ columns giving rise to the submatrix whose pseudo-inverse has the smallest norm). The rest of the proof is trivial and omitted for brevity.
\end{proof}

{\bf Greedy Algorithm for Question Selector.}
Under the linear model such as the one described in Equation~\ref{ls} and its equivalent representation using a pseudo-inverse matrix, the objective of identifying the set $\mathcal{Q}_s$ of $k$ selected questions (thereby identifying $m-k$ best factors) out of a set $\mathcal{Q}$ of $m$ questions (a task factor is a question) is equivalent to retaining the task factor submatrix with $m-k$ columns that is of the following form ~\cite{pukelsheim2006optimal}:

\begin{equation}\label{trace}
\argmin_{\mathcal{Q}_s \subset \mathcal{Q}, |\mathcal{Q}_s|=k} Trace({\mathcal{T}^f}^T_{Q\setminus Q_s} \mathcal{T}^f_{Q \setminus Q_s})^{-1}
\end{equation}

\begin{algorithm}
  \caption{Algorithm {\tt k-ExFactor}: Greedy Question Selector for a Linear Model
    \label{alg:bg}}
  \begin{algorithmic}[1]
    \Require{Task factor matrix $\mathcal{T}^f$, set of questions $\mathcal{Q}$}
    \Ensure{$\mathcal{Q}^s$ with $k$ factors}
    \Statex
      \Let{$\mathcal{T}_Q$}{$\mathcal{T}^f$}
      \Let{$\mathcal{Q}^s$}{$\mathcal{Q}$}
      \For{$j \gets 1 \textrm{ to } k$}
      \Let{$q_j$}{$argmin_{q \in \mathcal{Q}}$ Trace$({\mathcal{T}_Q}^T_{Q \setminus q} \mathcal{T}_{Q \setminus q})^{-1}$} \label{line:trace}
      \Let{$\mathcal{T}_Q$}{$\mathcal{T}_{Q \setminus j}$}
      \Let{$\mathcal{Q}^s$}{$\mathcal{Q} \setminus q_j$}
      \EndFor
       \State Return $\mathcal{Q} - \mathcal{Q}^s$
  \end{algorithmic}
\end{algorithm}

We now describe a greedy algorithm {\tt k-ExFactor} to identify $k$ worst task factors (thus retaining $m-k$ best factors). Our algorithm makes use of Equation~\ref{trace} and has a provable approximation guarantee. It works in a backward greedy manner and eliminates the factors iteratively. It works in $k$ iterations, and in the $i$-th iteration, from the not yet selected set of factors, it selects a question $q_j$ and eliminates it which marginally minimizes Trace$({\mathcal{T}^f}^T_{Q\setminus f_j} \mathcal{T}^f_{Q \setminus f_j})^{-1}$. Once the $k^{th}$ iteration completes the eliminated $k$ questions are the selected $k$-factors for explicit elicitation. The pseudo code of the algorithm is presented in Algorithm~\ref{alg:bg}. Line \ref{line:trace} in Algorithm~\ref{alg:bg} requires a $\mathcal{O}(m^2n+m^3)$ time for matrix multiplication and inversion for the question under consideration. Therefore, the overall complexity is $\mathcal{O}(km^2n^2+m^3nk)$. Notice that most of the complexity is actually in the process of recomputing the model error and the actual question selection is rather efficient.

\begin{theorem}
Algorithm {\tt k-ExFactor} has an approximation factor of $\frac{m}{m-k}$.
\end{theorem}

\begin{proof}
(sketch): The proof adapts from an existing result~\cite{de2007subset,avron2013faster} that uses backward greedy algorithm for {\em subset selection} for matrices and retains a given smaller number of columns such that the pseudo-inverse of the smaller sub-matrix has as smallest norm as possible. This is akin to removing $k$ worst task factors and retaining the best $m-k$ factors and the proof is a simple adaptation of~\cite{de2007subset,avron2013faster}.
 Exploration of a better approximation factor is deferred to future work.
\end{proof}

{\bf Running Example:} Using Example~\ref{ex1}, if $k=3$, $\{ \text{sentiment, Payoff, Duration}\}$ are the three task factors for which worker feedback is solicited.

\subsection{Preference Aggregator}\label{genp}
The second technical problem is to update $\mathcal{F}$ with the worker's preferences. Given a set of $k$ task factors, worker $w$ provides an absolute order on these factors. We design an algorithm that updates $\mathcal{F}$ using the same optimization function as the one used to build it initially, modified by adding the set of constraints that represent obtained preferences.


The worker provides a full order among the selected questions (task factors) in the form $i > j > r > l$. We express this full order using a set of pairwise constraints of the form $i > j$. If the preferences contain a full order among $k$ constraints, this gives rise to a total of $k(k-1)$ linear pairwise constraints.

Updating a Bayesian Network with constraints has been studied in the past~\cite{niculescu2006bayesian}. The idea is to add additional dummy variables that encode the constraints aptly. To satisfy $k(k-1)/2$ pairwise constraints, we will have to add that many number of dummy variables which considerably blows the size of the network. Unfortunately, such algorithms are expensive and unlikely to scale.

{\bf Efficient Solution for Preference Aggregator.}
For the simplified \textit{Worker Model}, with the linear constraints added to our objective function in Equation~\ref{ls}, the preference aggregation problem becomes a constrained least squares problem.

Specifically, our problem corresponds to a box-constrained least squares one as the solution vector must fall between known lower and upper bounds. The solution to this problem can be categorized into active-set or interior-point~\cite{mead2010least}. The active-set based methods construct a feasible region, compute the corresponding active-set, and use the variables in the active constraints to form an alternate formulation of a least squares optimization with equality constraints~\cite{stark1995bounded}. We use the interior-point method that is more scalable and encodes the convex set (of solutions) as a barrier function. It uses primal Newton Barrier method to ensure the KKT equality conditions to optimize the objective function~\cite{mead2010least}. The primal Newton Barrier interior-point is iterative and the exact complexity depends  on the barrier parameter and the number of iterations, but the algorithm is shown to be polynomial~\cite{stark1995bounded}.

{\bf Running Example:} Using Example~\ref{ex1} again, if the worker says that she prefers {\tt Duration > Sentiment > Payoff}, then the new weights that the preference aggregator estimates for $\mathcal{F}$ are, tagging=$0.1$, ranking= $0.1$, sentiment=$0.12$, payoff=$0.11$, duration=$0.97$. Notice that the order of the task factors provided by the worker is satisfied in the updated model.

\section{Experimental Evaluations}
\label{sec:experiments}
We describe our experimental setup, steps, and findings in this section. All the algorithms are implemented in Python 3.5.1. The experiments are conducted on a machine with Intel Core i7 4GHz CPU and 16GB of memory with Linux operating system. All the numbers are presented as an average of $10$ runs. We run both quality and scalability experiments and implement several baselines.

\subsection{Dataset Description}
We use $165,168$ CrowdFlower micro-tasks. A task belongs to one of the $22$ different categories, such as, tweet classification, searching information on the web, audio transcription, image tagging, sentiment analysis, entity resolution, or extracting information from news. Each task type is assigned a set of keywords that best describe its content and a payment, ranging between $\$0.01$ and $\$0.12$. Our tasks are {\em micro-tasks} that take less than a minute to complete.

Initially, we group a subset of micro-tasks into $240$ HITs and publish them on Amazon Mechanical Turk. Each HIT contains $20$ tasks and has a duration of $30$ minutes. When a worker accepts a HIT, he is redirected to our platform where he completes the tasks. A worker may complete several HITs in a work session. Workers get paid for every {\em micro-task} completed.

To qualify for our experiment, we require the workers to have previously completed at least $100$ HITs that are approved, and to have an approval rate above $80\%$. Overall, $58$ different workers complete tasks. When a worker is hired for the first time, she is asked to select a set of keywords from a given list of keywords that capture her preferences.

The task types along with other factors, such as, payment and duration, form the task factors. Our original data has $41$ task factors that are categorical or binary; after binarization of all the categorical factors, we obtain a total of $100$ factors. The length of any worker preference vector is therefore $100$. Of course, our proposed framework adapts even when the task factors are continuous.

{\bf Ground Truth.} Each micro-task has a known ground-truth. If a task, undertaken by a worker is completed successfully (i.e., the outcome matches the ground-truth), it is marked successful (value $1$). Otherwise, if the task is accepted but either not finished or not completed  correctly, its label is unsuccessful (value $0$). This information is used as the ground-truth for the \textit{Worker Model}.

{\bf Iteration.} We define an iteration as the completion of a HIT. When a worker finishes a HIT, we compute the error in the \textit{Worker Model}. We update the \textit{Worker Model}, when there is a non-zero error and start another iteration.

\subsection{Implemented Algorithms}
We now describe the algorithms that are implemented and compared for evaluation purposes.
\subsubsection{Worker Model \& Bootstrapping}
The linear model in Section~\ref{genmodel} is implemented with a regularization parameter $\alpha$. When implementing statistical models, this is a standard practice to avoid overfitting of the model. The overall objective function thus becomes,
\begin{equation}
\min_{\vec{w}^f\in\mathbb{R}^m} \left \|\mathcal{T}^O - \vec{w}^f  \times \mathcal{T}^f  \right \|_2 + \alpha \left \|\vec{w}^f  \right \|_2^2
\end{equation}
 The best value of $\alpha$ is chosen by generalized cross validation~\cite{mead2010least}.

Additionally, there are three algorithms that are implemented for bootstrapping the \textit{Worker Model}.

{\bf Random Bootstrapping.} {\tt RandomBoot} selects a random subset $\mathcal{B}$ of data as the initial tasks to present to the worker and records their outcome to estimate the \textit{Worker Model}.

{\bf Uniform Bootstrapping.} {\tt UniformBoot} does not learn anything to build the \textit{Worker Model} but bootstraps the model by assigning uniform weights to the worker preference vector.

{\bf Optimization-Aware Bootstrapping.} {\tt OptBoot} implements our algorithm given in Section~\ref{bootspec}.

\subsubsection{Explicit Feedback}

This has two important components - one is the {\bf Question Selector} that selects the task factors for explicit preference elicitation, the other is {\bf Preference Aggregator} that updates the \textit{Worker Model} using elicited preferences.\\
{\bf Optimization-Aware Question Selector.} {\tt k-ExFactor} is our proposed algorithm described in Section~\ref{genq}.\\
{\bf $k$-random Question Selector.} {\tt k-Random} is a simple baseline that randomly selects $k$-task factors for preference elicitation.
{\bf Preference Aggregator:} This is our implemented solution for preference aggregation, as described in Section~\ref{genp}.

\subsubsection{Implicit Feedback}
We also implement implicit feedback computation to serve as a comparison alternative to explicit feedback approaches.
Algorithm {\tt Implicit-1} is an adaption of a recent related work~\cite{edbtMata} that investigates how to implicitly capture worker motivation and use that for task assignment. While we do not necessarily focus on motivation as a factor in this work, we adapt the algorithm in~\cite{edbtMata} to estimate and update the worker preference vector over time. Since our focus is not on task assignment, once we estimate the worker preference vector using {\tt Implicit-1}, we use that in conjunction with our \textit{Worker Model} to predict a task outcome.
Algorithm {\tt Implicit-2} is a further simplification. It relearns the \textit{Worker Model} at the end of every iteration as the worker completes tasks and does not factor in the preference of the worker in the \textit{Worker Model}.

\subsection{Summary of Results}
There are two primary takeaways: \\
 {\bf 1. Our proposed explicit preference elicitation framework outperforms existing implicit ones with statistical significance.} We compare our approach {\tt k-ExFactor} for question selection with another explicit baseline {\tt k-Random} and two implicit preference computation algorithms {\tt Implicit-1}~\cite{edbtMata} and {\tt Implicit-2}. For qualitative evaluation, we present error (mean square error or MSE) with statistical significance test and find that {\tt k-ExFactor} convincingly and significantly outperforms the other three baselines under varying parameters: \# iterations, \# task factors, $k$. For bootstrapping, we compare our  solution {\tt OptBoot} with two baselines {\tt UniformBoot} and {\tt RandomBoot}. Again, our observation here is {\tt OptBoot} is superior qualitatively.\\
{2. \bf Our proposed solution is scalable.} In our scalability study, we vary \# tasks, \# task factors, $k$, and bootstrap sample size. While {\tt k-ExFactor} is slower than the other three algorithms, as it performs a significantly higher number of computations, it still scales very well. Unsurprisingly, our bootstrapping algorithm {\tt OptBoot} is slower than the two other baselines. Despite that, it scales reasonably well. These results demonstrate the effectiveness of eliciting explicit preferences making our framework usable in practice.

\subsection{Quality Experiments}
\label{sec:qual}
The objective of these experiments is to capture the effectiveness of our  explicit feedback elicitation framework and compare it with appropriate baselines. Unless otherwise stated, we capture effectiveness as reconstruction error, according to the objective function in Section~\ref{prob} using Mean Square Error or MSE.

{\bf Invocation of the Framework.} For quality experiments, the proposed framework is invoked iteratively as follows: in the beginning, we filter out the tasks and task completion history by worker id since the framework is personalized per worker. On average, a worker undertakes $100$ tasks. We randomly sample $70\%$ of each worker's data for training and the rest as the holdout set.

For bootstrapping, the \textit{Worker Model} is initially built by selecting a subset $\mathcal{B}$ of $b$ tasks from the training data and MSE is computed over the holdout set. After that, in an iteration, we select a set $x$ of $25$ tasks (unless otherwise stated), randomly from the holdout set and we calculate the score over the remaining set of tasks in the holdout set. Next, we will check if there is an error in the prediction of \textit{Worker Model} for set $x$. If yes, then we invoke the {\bf Question Selector} that seeks explicit feedback from the same worker. Upon receiving worker feedback, the \textit{Worker Model} is updated using the {\bf Preference Aggregator}. All these steps construe a single iteration of the framework. We periodically perform the aforementioned steps to get multiple iterations of the framework.

{\bf Parameter Setting.} For a given worker, there are three parameters to vary:  \# task factors, $k$, and \# iterations. For bootstrapping, we additionally vary the budget $b$. Unless otherwise stated, defaults values are $90$, $4$, and $7$, respectively. The best $90$ features are retained by performing feature selection using {\em Chi-squared test}\cite{forman2003extensive}. We also notice that the error does not reduce significantly beyond $k=4$ questions and $7$ iterations. By default, we always maintain the full history of worker preference while updating the \textit{Worker Model} under varying iterations and the default size of the bootstrapping set is $15$.

\subsubsection{Explicit vs. Implicit Feedback}
Figure~\ref{expimp} presents a comparative study between explicit, implicit, and no preference elicitation. We compare two explicit solutions with two implicit   ones. We vary \# iterations, \# task factors, and $x$ (\# tasks assigned to a worker after which the framework is invoked).
Figure\ref{Model_result} presents the error of each of the four algorithms by varying the number of iterations, where we compare two explicit algorithms ((\texttt{k-ExFactor} and \texttt{k-Random}) with implicit ones \texttt{Implicit-1}~\cite{edbtMata} and \texttt{Implicit-2}. Our method \texttt{k-ExFactor} significantly outperforms the other three. After $7$ iterations its error drops from $41\%$ to $15\%$ almost $10\%$ lower than other methods. A similar observation holds for \texttt{k-ExFactor} when we vary \# task factors (Figure~\ref{Model_feature}), and \# tasks (Figure~\ref{Model_tasks}). Our proposed solution convincingly and significantly outperforms \texttt{k-Random} and implicit preference computation.

\begin{figure*}[ht]
\centering
\subfigure[Error - varying iterations]{
   \includegraphics[height=4.2cm, width=5cm]{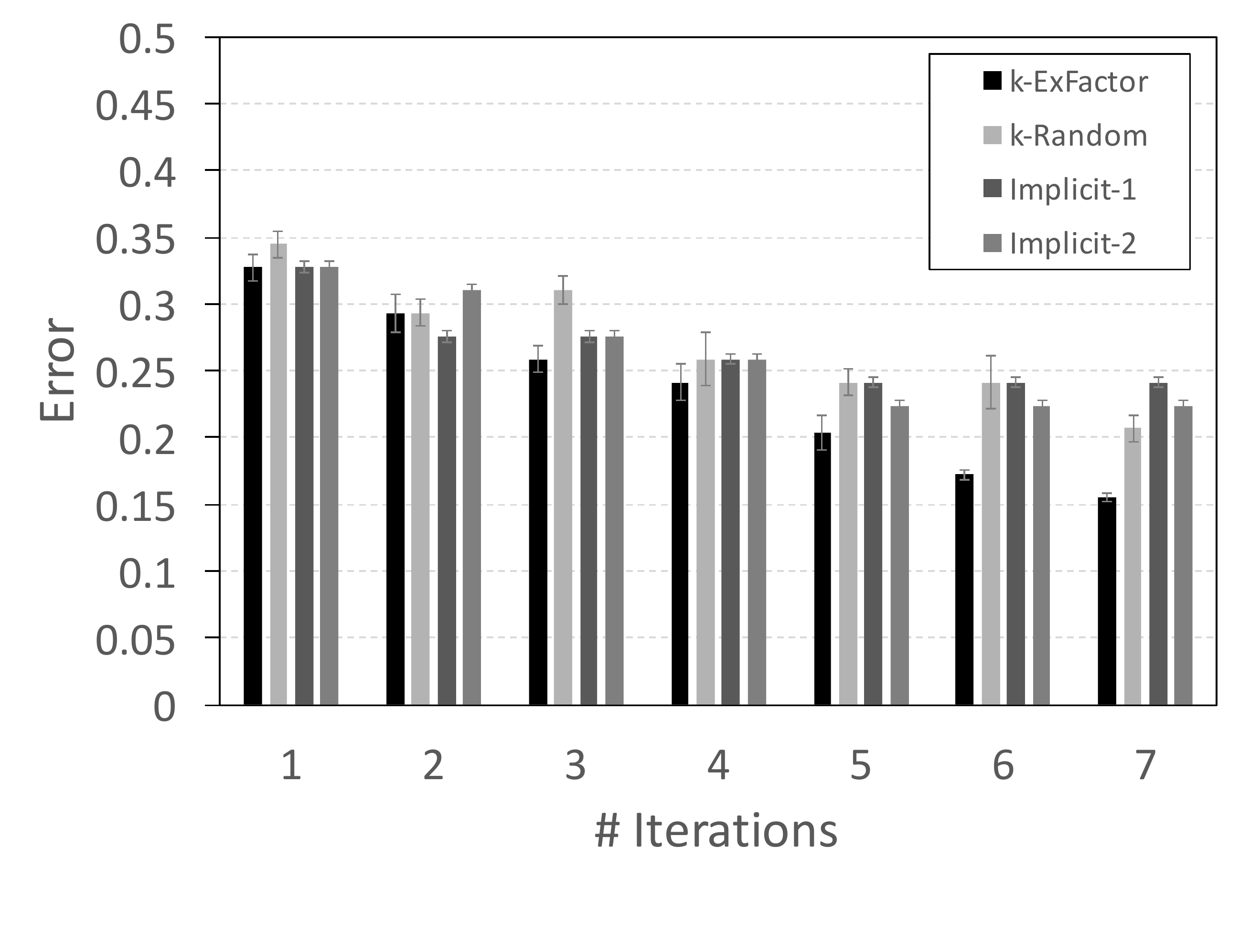}
     \label{Model_result}
    }
\subfigure[Error - varying task factors]{
   \includegraphics[height=4cm, width=5cm]{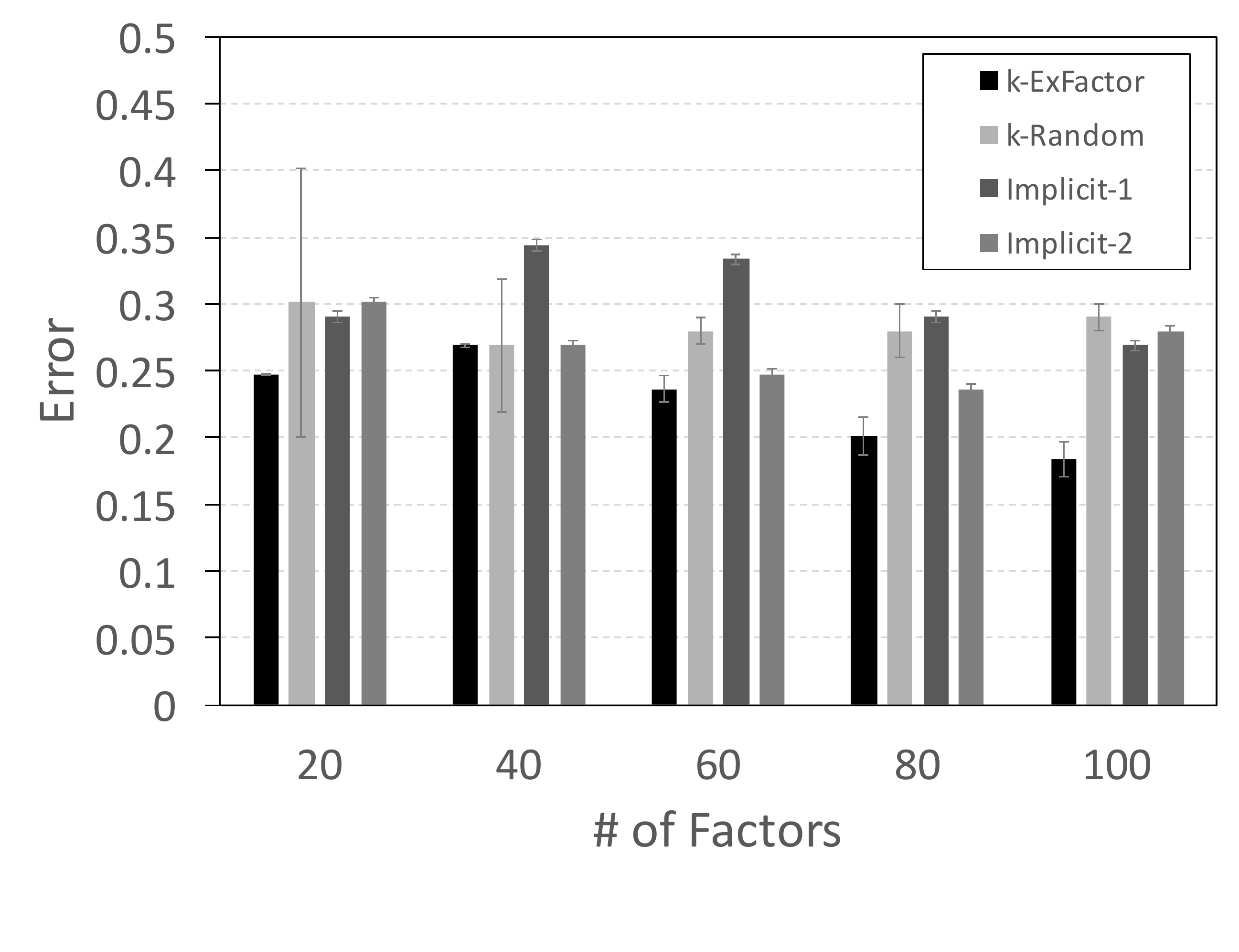}
     \label{Model_feature}
    }
\subfigure[Error - varying $x$]{
   \includegraphics[height=4cm, width=5cm]{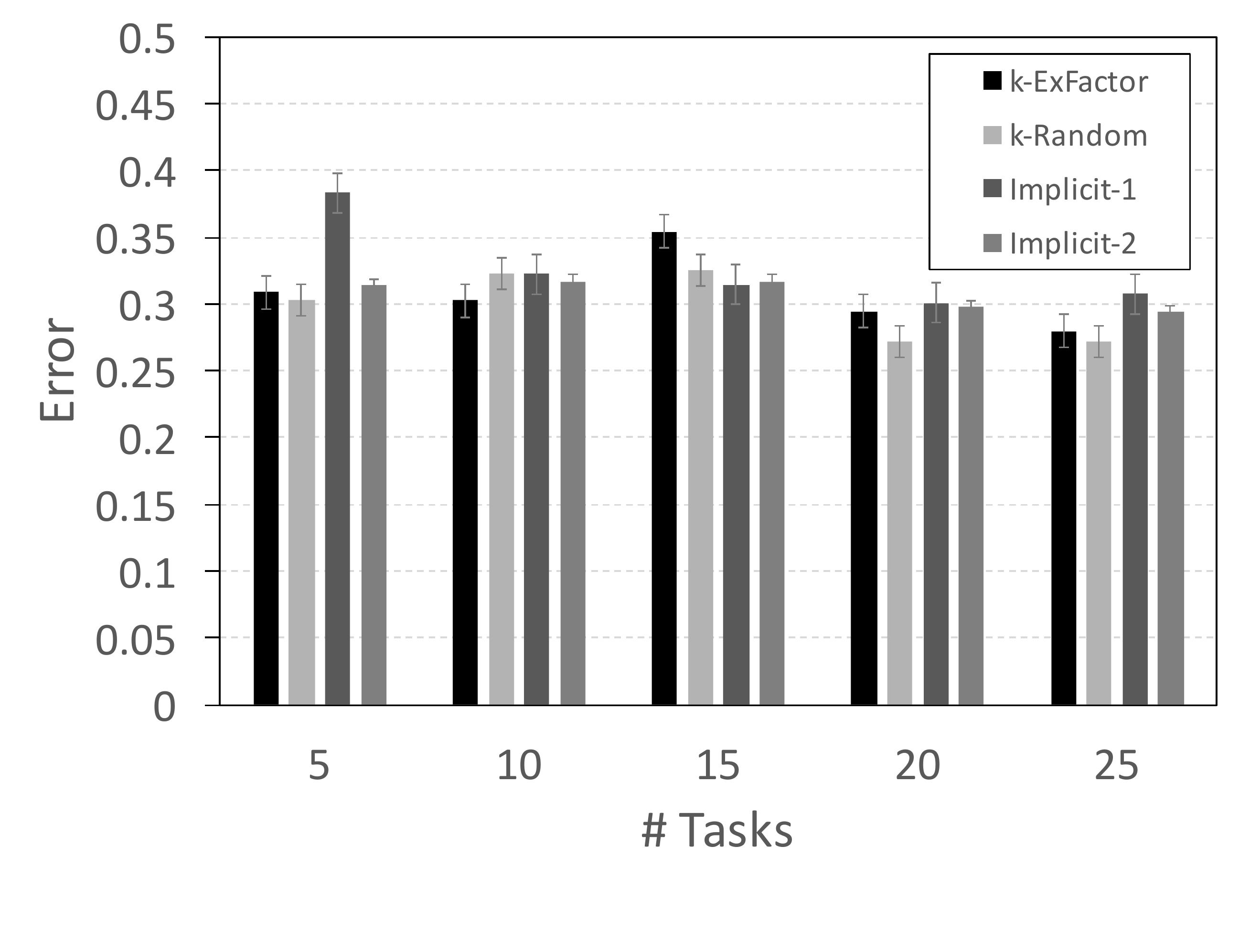}
     \label{Model_tasks}
    }
\vspace{-0.1in}
\caption{Comparison between Explicit and Implicit Preferences with Statistical Significance Test (standard error) \label{expimp}}
\end{figure*}

\subsubsection{Explicit Feedback}
Figure \ref{k_factor_result} presents the error of the two explicit preference elicitation methods as a function of the number of questions $k$. Notice that the two implicit preference algorithms do not have an input parameter $k$ but for the sake of comparison we have include their results in Figure \ref{k_factor_result}. Overall, our method, \texttt{k-ExFactor}, clearly outperforms \texttt{k-Random} and the other two implicit methods. More importantly, increasing the number of questions does not necessarily yield better results as it is shown in Figure~\ref{k_factor_result}. This could be justified by the fact that adding more constraints to the model will result in poor optimization results. That indicates that a small number of questions is good enough to elicit worker preferences and improve the \textit{Worker Model}.
\begin{figure}
  \centering
      \includegraphics[width=0.4\textwidth]{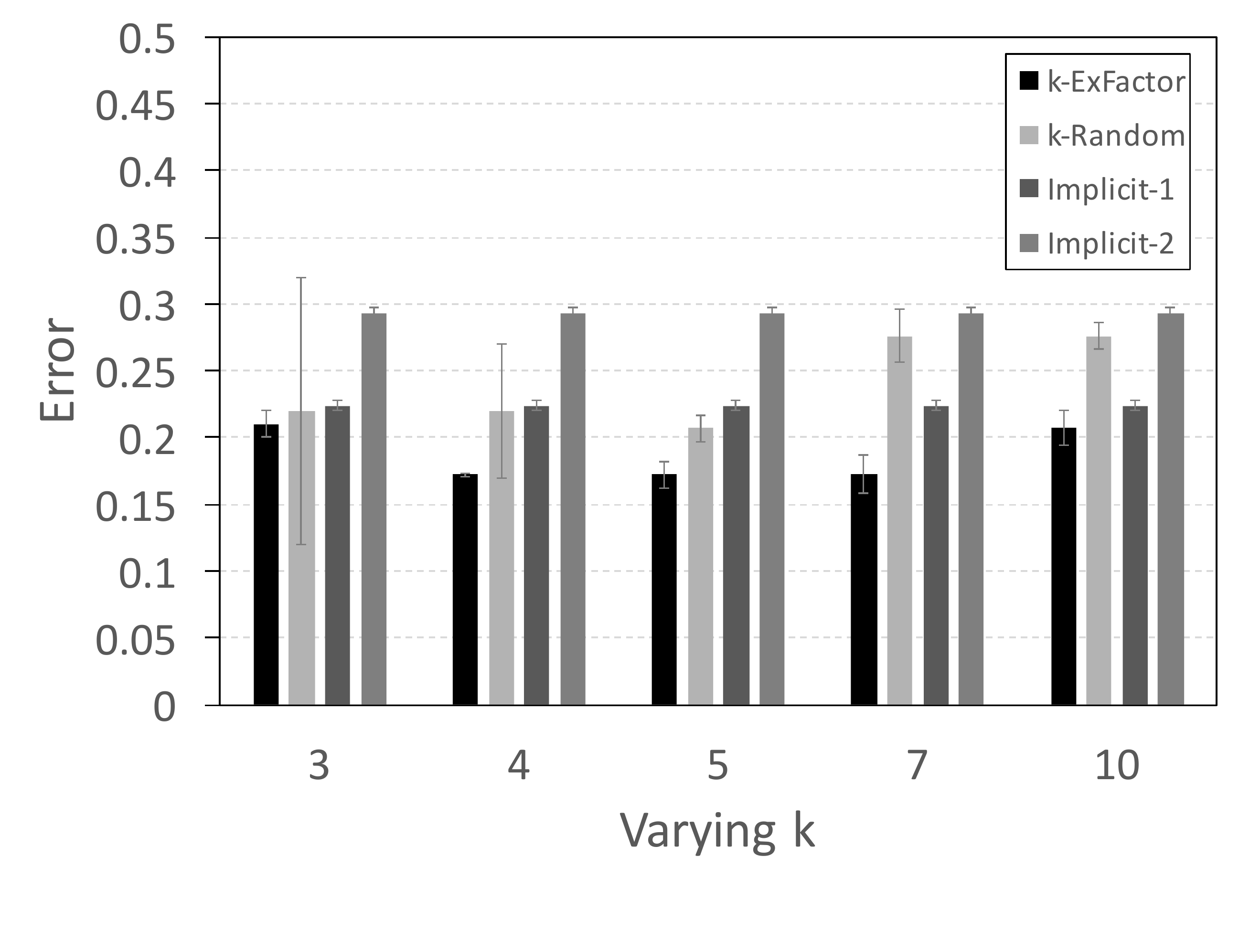}
  \caption{Error varying $k$ with Statistical Significance Test (standard error)}
  \label{k_factor_result}
\end{figure}

\subsubsection{Explicit Feedback: Full History vs Partial History}
We now present a comparative study between capturing the full history of worker preferences vs just the most recent preferences in the preference aggregation step. The size of history is the total number of explicit feedbacks received from the workers from beginning until a given point in time. As an example, if we ask $4$ explicit questions to a worker in each iteration, after the second iteration, her full history size is $8$, whereas, her most recent history size is $4$; i.e., the recent history represents the number of feedbacks in the current iterations. Figure~\ref{history_result} demonstrates the results by varying \# iterations. Clearly, the rate of error reduction for the model with a full history is higher and the error of the model in the end is slightly smaller than the model with just the most recent history. Taking into account the evolution of worker preferences in the whole session is therefore a better option.

\begin{figure}
  \centering
      \includegraphics[width=0.4\textwidth]{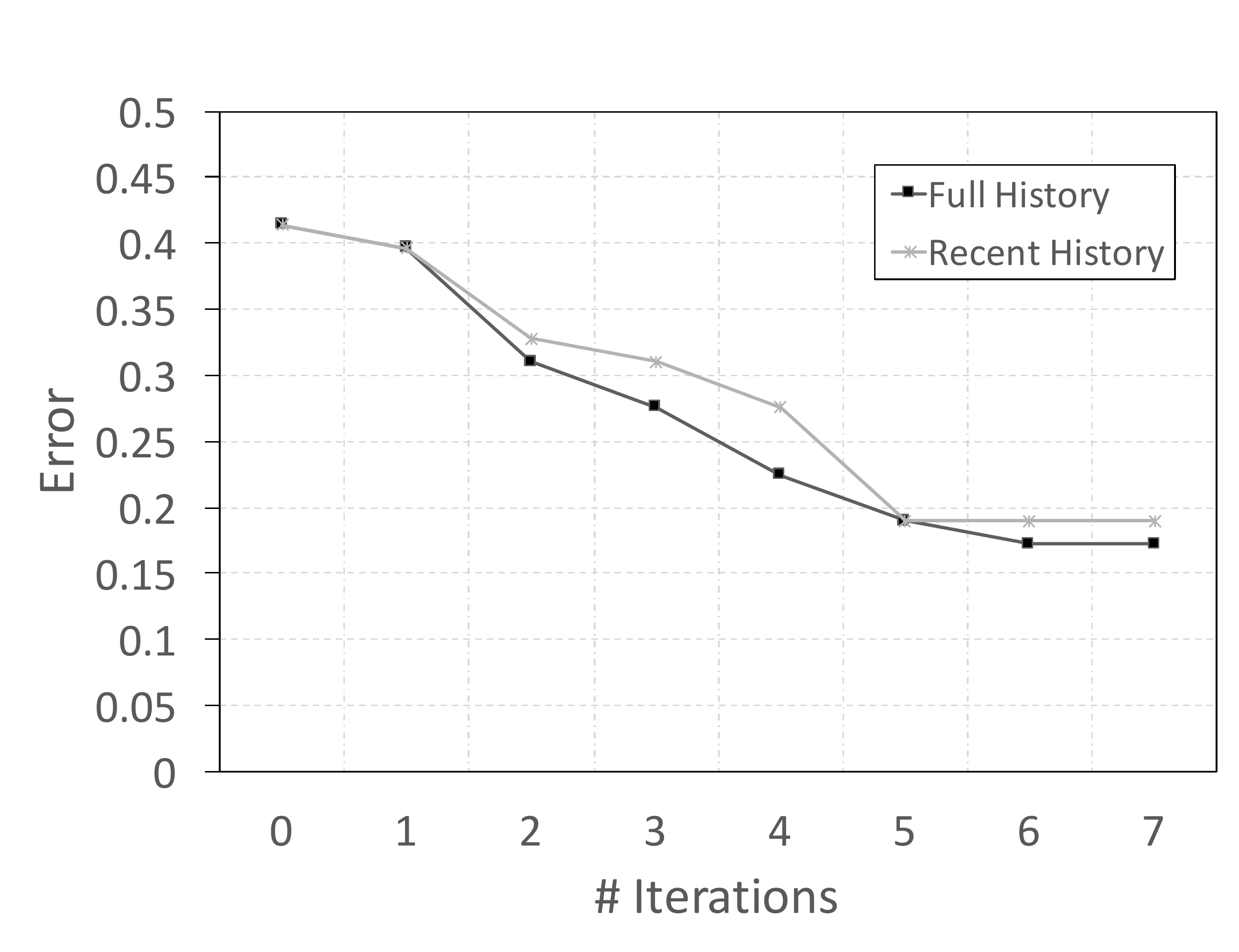}
  \caption{Comparison between full and recent history}
  \label{history_result}
\end{figure}

\subsubsection{Bootstrapping}
Figure \ref{bootstrap_result} presents the error of the three bootstrapping algorithms. For \texttt{RandomBoot} and \texttt{OptBoot} we set $b=15$ tasks, whereas, {\tt UniformBoot} just sets uniform weights to the worker preference vector. We continue to add an additional number of $b=15$ tasks from the training set and measure MSE. Initially, \texttt{OptBoot} has the best error which signals the effectiveness of our method to pick the best set of tasks that gives us the smallest error. In the preceding iterations, \texttt{OptBoot} converges to a lower error faster than the other two methods. The rate of decrease in error is the same after the fourth iteration which shows that the \textit{Worker Model} is stable and performs well.

\begin{figure}
  \centering
      \includegraphics[width=0.4\textwidth]{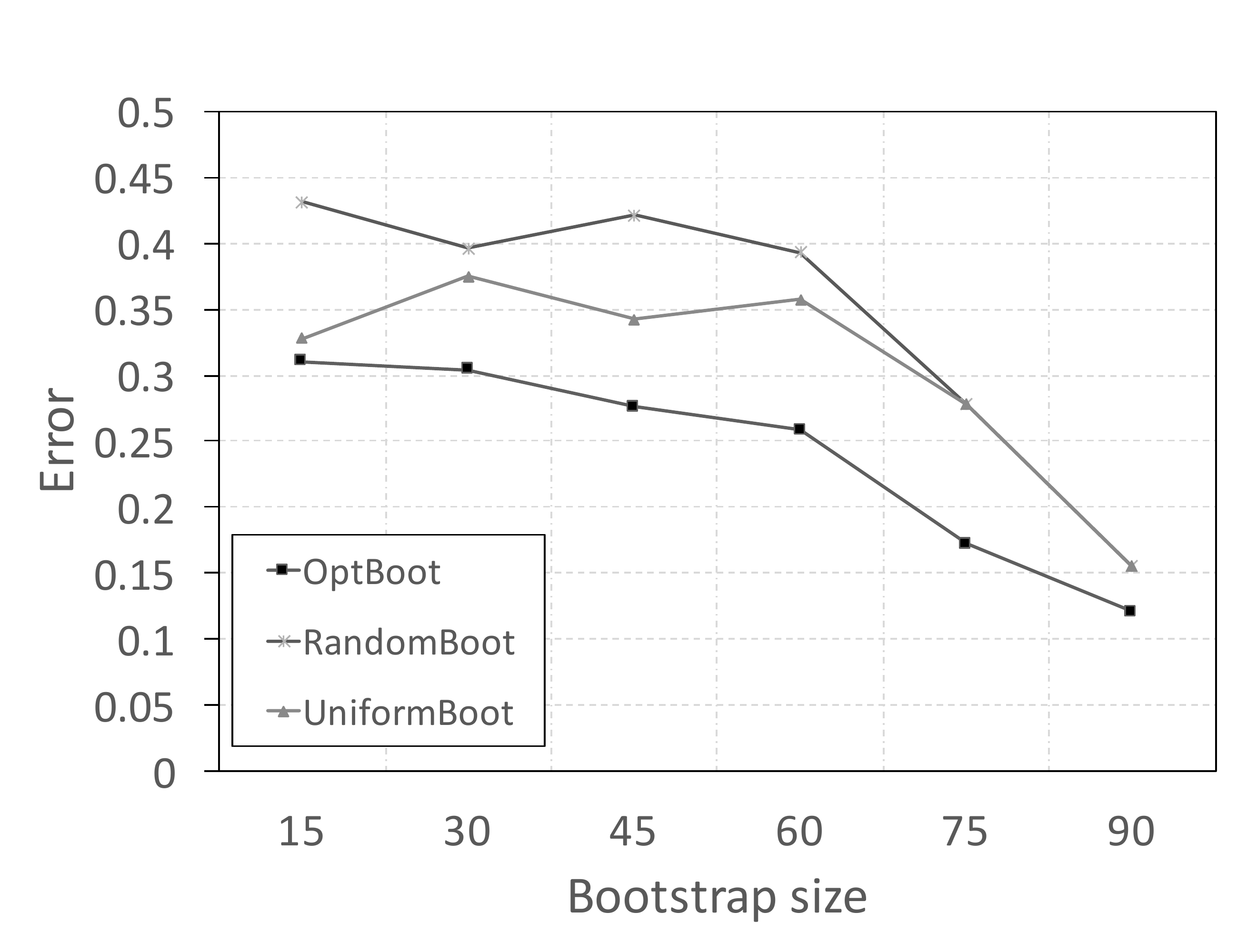}
  \caption{Evaluation of bootstrapping algorithms}
  \label{bootstrap_result}

\end{figure}

\subsubsection{Worker Model}
We profiled three workers randomly from our database and analyzed their models in conjunction with the keywords they have initially chosen. Table~\ref{tab:prof} presents the $6$ keywords chosen by the workers and the top-$2$ worker preferences. It is easy to notice that they are highly correlated, which shows that our proposed model successfully captures worker preference.

\begin{table}
    \begin{tabular}{ | l | p{3cm} | p{3cm} | }
    \hline
    Worker no & Worker Keywords & Top-$2$ preference \\ \hline
    1 & dress,google street view, airlines, classification, wheelchair accessibility, scene & dress, scene \\ \hline
    2 & business, body parts, google street view, health, new year resolution, classification & classification, google street view  \\ \hline
    3 & image, south Asia, disease, animals, text & image, text \\ \hline
    \end{tabular}\caption{Worker Keywords and Worker Model}\label{tab:prof}
\end{table}

\subsection{Scalability Experiments}
\label{sec:eff}
We conduct an in-depth scalability study of our solutions and their competitors. Unless otherwise stated, we always report running time in seconds.

{\bf Parameter Setting.} Our dataset contains $165,168$ tasks and $100$ task factors obtained from $58$ workers. In this experiment, we vary the following parameters: \# tasks, \# task factors, $k$, and the bootstrapping budget $b$. Unless otherwise stated, all the numbers present the average running time of a single iteration over all the 58 workers. The default values are set as
\# tasks = $50,000$, \# task factors = $50$, $k=3$, and $b=25$. Unless otherwise stated, all four algorithms are compared with each other. For bootstrapping comparison, only the appropriate three methods are compared.

Figure~\ref{scale} presents the results. Figure~\ref{tasks} presents the running time of the four algorithms with varying number of tasks. Of course, our proposed solution {\tt k-ExFactor} makes a lot more computation to ensure optimization and hence has the highest running time. However, it is easy to notice that with an increasing number of tasks, it scales well and the running time is comparable to the other competing algorithms. A similar observation holds when we vary the number of task factors, as shown in Figure~\ref{feat}. {\tt k-ExFactor} scales well and never takes more than $86$ seconds. Figure~\ref{varyk} represents the running times by varying $k$, the number of task factors chosen for preference elicitation.  Here only {\tt k-ExFactor} is compared with {\tt k-Random}, as the other two algorithms do not rely on explicit preference elicitation. Unsurprisingly, {\tt k-Random} is faster, but our proposed solution {\tt k-ExFactor} scales well and has a comparable running time. Finally, in Figure~\ref{varyb}, we vary the bootstrapping sample size and present the running time of {\tt OptBoot}. For efficient implementation, we only randomly profile $10\%$ of the branches of the bootstrap tree which makes the algorithm scale linearly. The other two baselines basically do not involve  any computation and take negligible time to terminate.

\begin{figure*}[ht]
\centering
\subfigure[varying \# tasks]{
   \includegraphics[height=4cm, width=4cm]{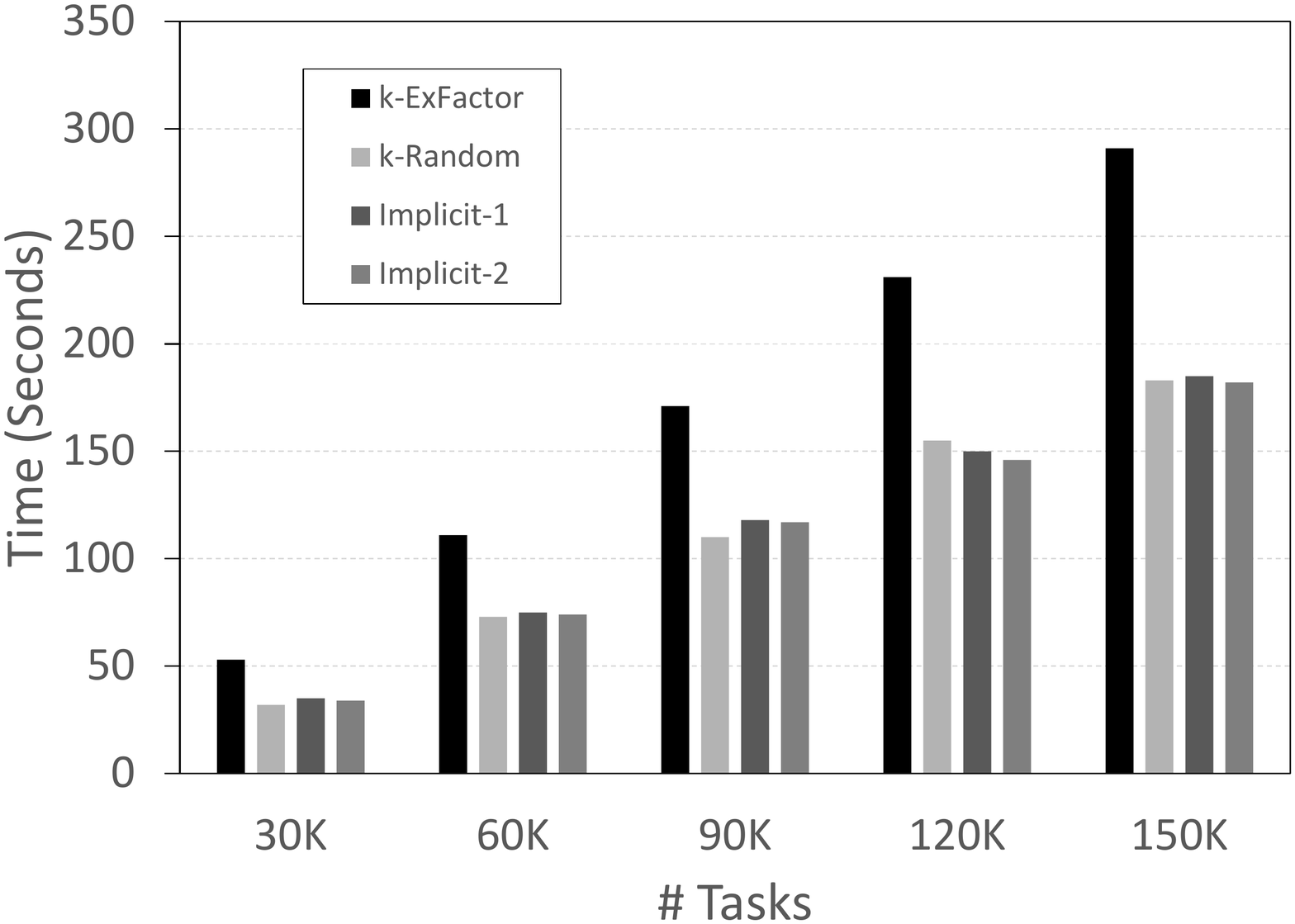}
     \label{tasks}
    }
\subfigure[varying \# task factors]{
   \includegraphics[height=4cm, width=4cm]{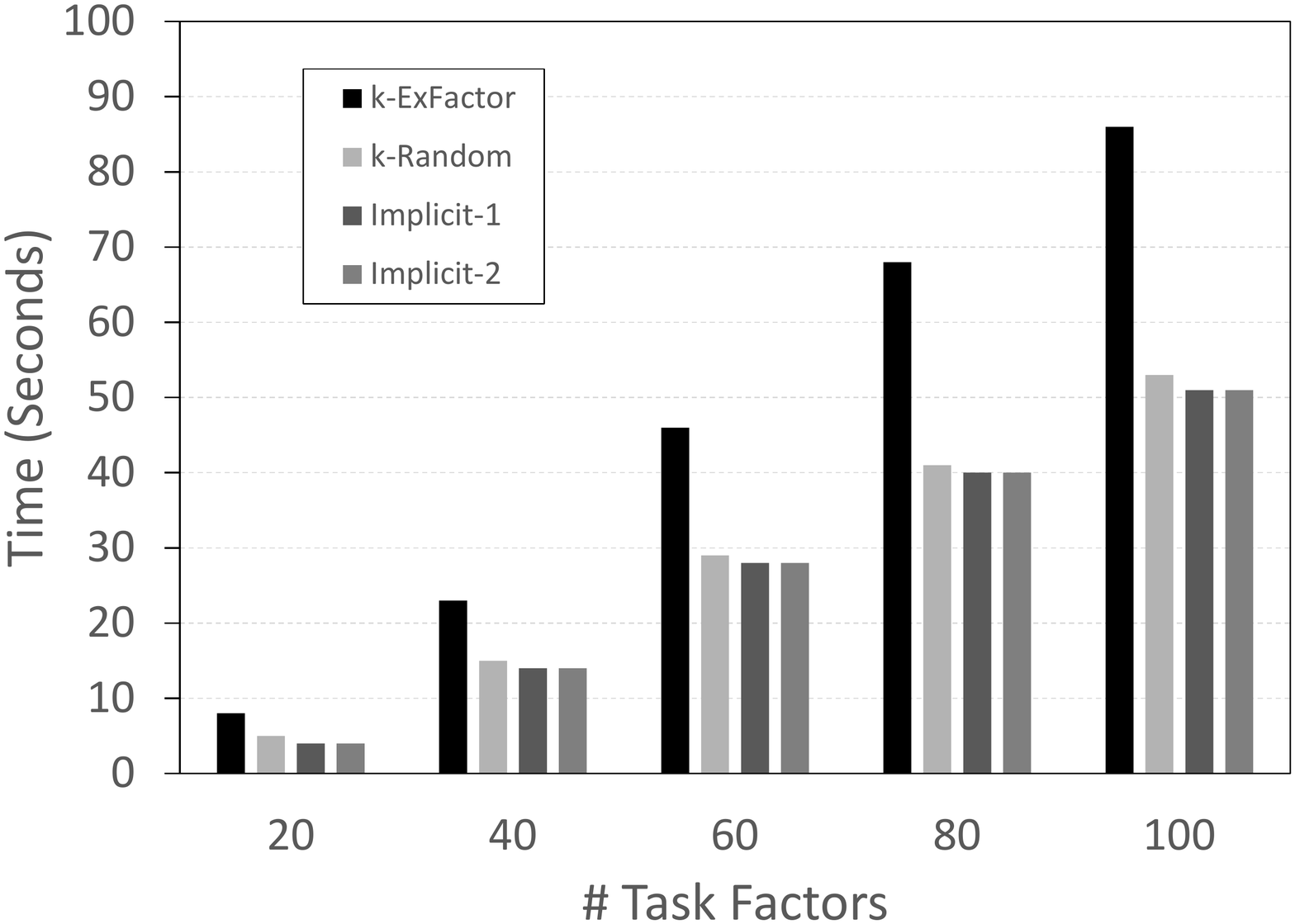}
     \label{feat}
    }
\subfigure[varying $k$]{
   \includegraphics[height=4cm, width=4cm]{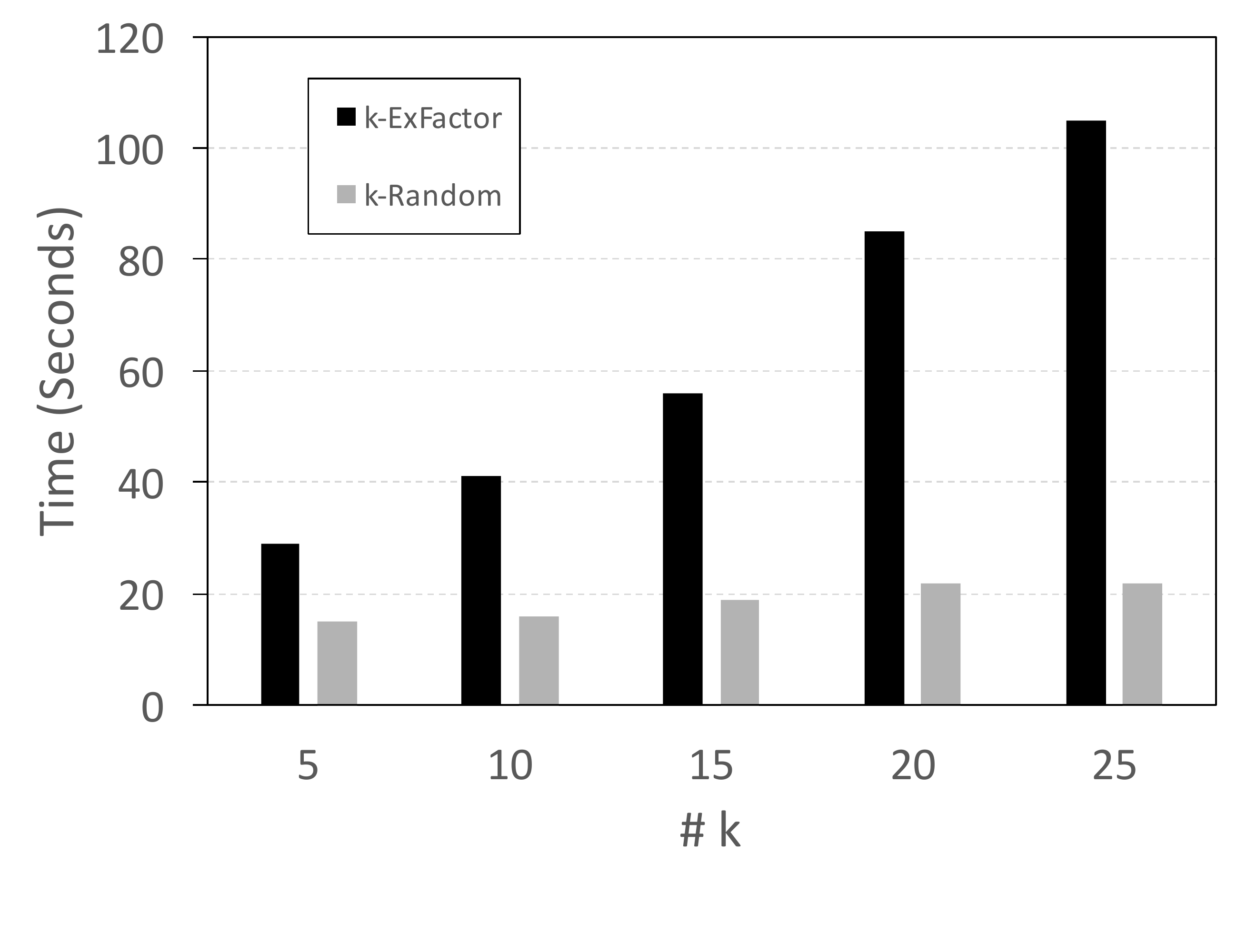}
     \label{varyk}
    }
    \subfigure[varying $b$]{
   \includegraphics[height=4cm, width=4cm]{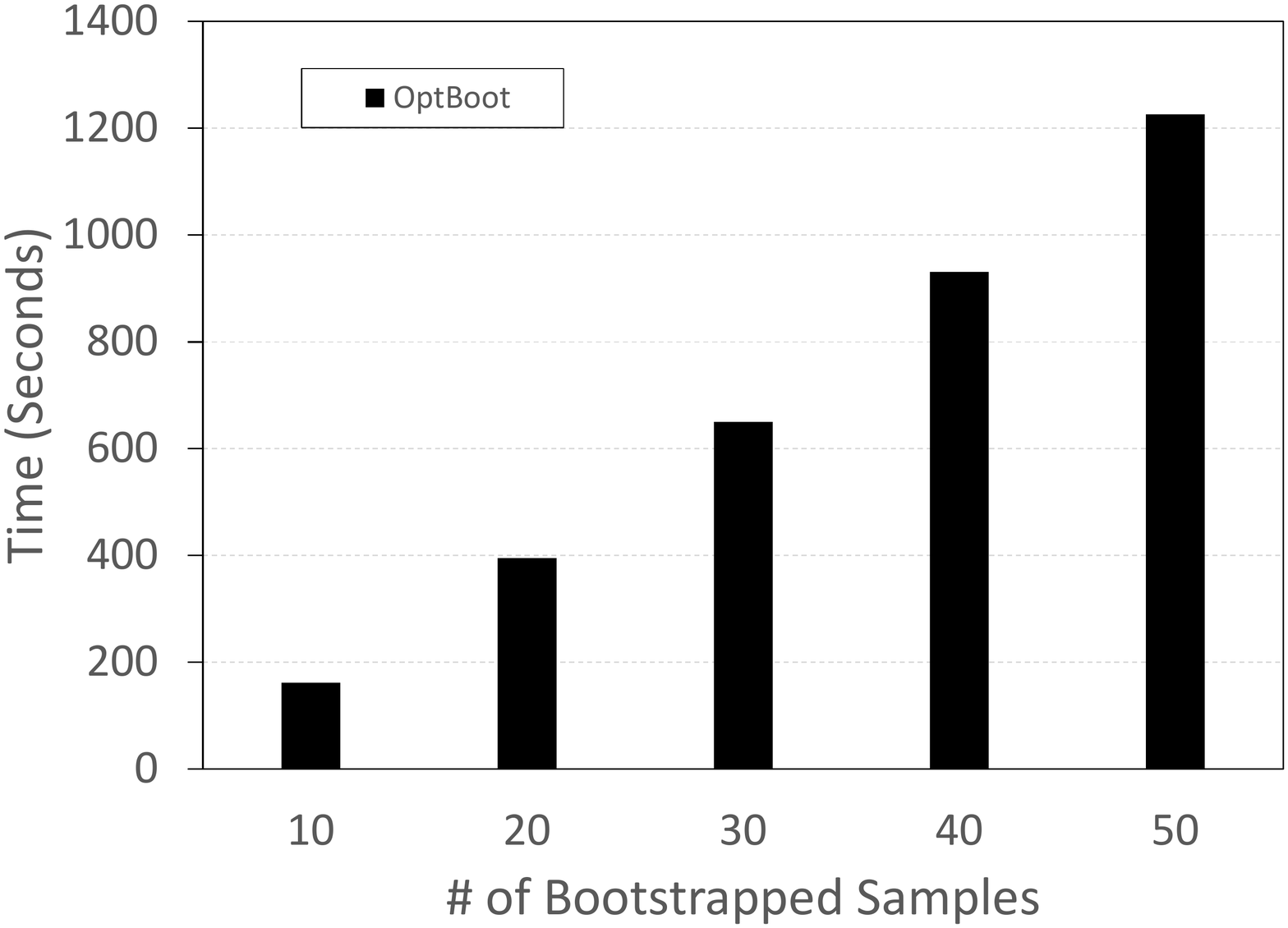}
     \label{varyb}
    }
\vspace{-0.1in}
\caption{\small Scalability study\label{scale}}
\end{figure*}

{\bf Profiling {\tt k-ExFactor}.}
We further profile the individual running time of {\tt k-ExFactor} with the default settings; i.e., \# tasks = $50,000$, \# task factors= $50$, $k=3$. It takes $28$ seconds to train the \textit{Worker Model}, $35.85$ seconds to solve {\bf Question Selector} that finds the best $k$ factors, and $29.1$ seconds to run {\bf Preference Aggregation} that updates the \textit{Worker Model} with the added constraints. These results demonstrate that the individual components of the framework take comparable time.

\section{Related Work}
\label{sec:relatedwork}
The related work can be classified into three categories: preference elicitation from the crowd, leveraging worker preferences in crowdsourcing processes, and worker models.\\  

{\bf Preference Elicitation.}
In~\cite{DBLP:conf/sigmod/GuoPG12,DBLP:conf/webdb/PolychronopoulosADGP13,DBLP:journals/tods/DavidsonKMR14}, the crowd was solicited to perform max/top-k and clustering operations with the assumption that workers may make errors. These papers study the relationship between the number of comparisons needed and error. Efficient algorithms are proposed with a guarantee to achieve correct results with high probability. A similar problem was addressed in~\cite{DBLP:conf/pods/GrozM15} in the case of a skyline evaluation. In that setting, it is assumed that items can only be compared through noisy comparisons provided by the crowd and the goal is to minimize the number of comparisons. 
A recent work studies the problem of computing the all pair distance graph~\cite{DBLP:conf/edbt/2017} by relying on noisy human workers. The authors addressed the challenge of how to aggregate those feedback and what additional feedback to solicit from the crowd to improve other estimated distances.

\textit{While we also rely on inputs from the crowd, the elicited input represents each worker's preference for different factors (as opposed to completing actual tasks), and is hence not assumed to be noisy or erroneous. However, as worker preferences evolve over time, we propose an iterative approach with the goal of improving task completion overall.}\\

{\bf Leveraging Preferences.}
Worker preferences for task factors are heavily leveraged in all crowdsourcing processes. 
Very few of these efforts focused on leveraging them in {\em task completion}~\cite{DBLP:journals/corr/abs-1210-0962,DBLP:conf/cscw/DaiRPC15,DBLP:conf/cscw/ShawHC11}. Authors of~\cite{DBLP:conf/amcis/KaufmannSV11} investigated $13$ worker {\em motivation} factors  and found that workers were interested in \textit{skill variety} or \textit{task autonomy} as much as {\em task reward}. Chandler and Kapelner~\cite{DBLP:journals/corr/abs-1210-0962} empirically showed that workers {\em perceived meaningfulness} of a task improved throughput without degrading quality. Shaw et al.~\cite{DBLP:conf/cscw/ShawHC11} assessed 14 incentives schemes and found that incentives based on {\em worker-to-worker comparisons} yield better crowd work quality. Hata et al.~\cite{CSCW17} studied worker {\em fatigue} and it affects how work quality over extended periods of time. Other efforts focused on gradually increasing pay during task completion to improve worker retention~\cite{GaoP14}. Lately, adaptive task assignment were studied with a particular focus on maximizing the quality of crowdwork~\cite{fan2015icrowd,HoJV13,HoV12,ZhengWLCF15,edbtMata} but primary for improved task assignment. 

{\em Existing work showed the importance of leveraging implicit worker preferences for task assignment. In contrast, we show explicit elicitation of worker preferences results in a more accurate model that leads to better task completion.}\\

{\bf Worker Model.}
Matrix factorization~\cite{ece,us} is used to recommend tasks to workers, where both worker and task features are latent variables in a lower dimensional space. In our work, task factors are explicit and known since they are provided by the crowdsourcing platform and by requesters. Further complex models, such as  Multi-Layered Networks as the ones used in deep learning, or Bayesian model are possible but are hard to scale - thus the two core computational problems in our framework that have to excessively use these models become prohibitively complex. \\
{\em To the best of our knowledge, we present the first principled solution for explicit preference elicitation and rigorously study scalability.}

\section{Conclusion and Future Work}
\label{sec:conclusion}
We initiate the study of investigating explicit preference elicitation for improved task completion. Our proposed framework leverages a \textit{Worker Model} that is personalized and learns from the past history of the worker and the task characteristics to predict task outcome. Two central problems that are part of our framework are {\bf Question Selector} and {\bf Preference Aggregator}. The former selects the best set of questions to elicit explicit preferences and the latter updates \textit{Worker Model} with the obtained preferences. We present principled solutions and experimentally validate the effectiveness of explicit preference elicitation. Next, we discuss some interesting extensions.

{\bf Combining Explicit and Implicit Preference.} An interesting open question is how to combine explicit preference elicitation with implicit preference computation. Our current understanding of the problem is, we only invoke explicit preference when certain event triggers it: such as, the error of the {\em Worker Model} obtained through implicit preference is too high, the worker is not undertaking enough tasks, the implicit preference solution is unable to discriminate worker's preference sufficiently. The challenge of this new problem is to design an optimization function that will guide when to seek explicit preference and when not to. 

{\bf Handling Multiple Workers.} A natural extension of our studied framework is to build and maintain a \textit{Worker Model} {\em not per worker, but for a set of workers}. A simple approach would cluster workers based on their preference vectors and aggregate the individual \textit{Worker Models} to build a \textit{Virtual Worker Model}. Such a model is likely to introduce more error (as the model is no longer personalized per worker) but is going to be more efficient to maintain. Such a model could further be used to profile workers in crowdsourcing platforms or improve crowdsourcing processes such as recruitment, completion or assignment.

{\bf Worker Models.} We present a supervised approach to develop solutions for the \textit{Worker Model}. A natural alternative is to study this problem in an unsupervised setting where the worker history is not available, using techniques such as Self Organizing Maps~\cite{kohonen1998self}.A further interesting extension is in removing the assumption that there is an one-to-one correspondence between task factors and explicit questions. As long as the correspondence between the task factors and explicit questions is defined, our proposed optimization framework would be adapted.

\bibliographystyle{IEEEtran}

\begin{thebibliography}{10}
\providecommand{\url}[1]{#1}
\csname url@samestyle\endcsname
\providecommand{\newblock}{\relax}
\providecommand{\bibinfo}[2]{#2}
\providecommand{\BIBentrySTDinterwordspacing}{\spaceskip=0pt\relax}
\providecommand{\BIBentryALTinterwordstretchfactor}{4}
\providecommand{\BIBentryALTinterwordspacing}{\spaceskip=\fontdimen2\font plus
\BIBentryALTinterwordstretchfactor\fontdimen3\font minus
  \fontdimen4\font\relax}
\providecommand{\BIBforeignlanguage}[2]{{%
\expandafter\ifx\csname l@#1\endcsname\relax
\typeout{** WARNING: IEEEtran.bst: No hyphenation pattern has been}%
\typeout{** loaded for the language `#1'. Using the pattern for}%
\typeout{** the default language instead.}%
\else
\language=\csname l@#1\endcsname
\fi
#2}}
\providecommand{\BIBdecl}{\relax}
\BIBdecl

\bibitem{BedersonQ11}
B.~B. Bederson and A.~J. Quinn, ``Web workers unite! addressing challenges of
  online laborers,'' in \emph{{CHI}}, 2011, pp. 97--106.

\bibitem{KitturNBGSZLH13}
A.~Kittur \emph{et~al.}, ``The future of crowd work,'' in \emph{{CSCW}}, 2013.

\bibitem{DBLP:conf/dbcrowd/RoyLTAD13}
S.~B. Roy \emph{et~al.}, ``Crowds, not drones: Modeling human factors in
  interactive crowdsourcing,'' in \emph{DBCrowd}, 2013.

\bibitem{amer2016human}
S.~Amer-Yahia and S.~B. Roy, ``Human factors in crowdsourcing,''
  \emph{Proceedings of the VLDB Endowment}, 2016.

\bibitem{kaufmann2011more}
N.~Kaufmann, T.~Schulze, and D.~Veit, ``More than fun and money. worker
  motivation in crowdsourcing-a study on mechanical turk.'' in \emph{AMCIS},
  2011.

\bibitem{DBLP:journals/corr/abs-1210-0962}
D.~Chandler and A.~Kapelner, ``Breaking monotony with meaning: Motivation in
  crowdsourcing markets,'' \emph{CoRR}, vol. abs/1210.0962, 2012.

\bibitem{DBLP:conf/sigecom/HortonC10}
J.~J. Horton and L.~B. Chilton, ``The labor economics of paid crowdsourcing,''
  in \emph{ACM EC}, 2010, pp. 209--218.

\bibitem{MartinHOG14}
D.~B. Martin \emph{et~al.}, ``Being a turker,'' in \emph{{CSCW}}, 2014.

\bibitem{RogstadiusKKSLV11}
J.~Rogstadius \emph{et~al.}, ``An assessment of intrinsic and extrinsic
  motivation on task performance in crowdsourcing markets,'' in \emph{ICWSM},
  2011.

\bibitem{CSCW17}
K.~Hata \emph{et~al.}, ``A glimpse far into the future: Understanding long-term
  crowd worker quality,'' in \emph{CSCW}, 2017.

\bibitem{DBLP:conf/cscw/DaiRPC15}
P.~Dai \emph{et~al.}, ``And now for something completely different: Improving
  crowdsourcing workflows with micro-diversions,'' in \emph{ACM CSCW}, 2015.

\bibitem{koller2009probabilistic}
D.~Koller and N.~Friedman, \emph{Probabilistic graphical models: principles and
  techniques}.\hskip 1em plus 0.5em minus 0.4em\relax MIT press, 2009.

\bibitem{edbtMata}
J.~Pilourdault \emph{et~al.}, ``Motivation-aware task assignment in
  crowdsourcing,'' in \emph{EDBT}, 2017.

\bibitem{ZhengWLCF15}
Y.~Zheng \emph{et~al.}, ``{QASCA:} {A} quality-aware task assignment system for
  crowdsourcing applications,'' in \emph{{SIGMOD}}, 2015.

\bibitem{HoV12}
C.~Ho and J.~W. Vaughan, ``Online task assignment in crowdsourcing markets,''
  in \emph{{AAAI}}, 2012.

\bibitem{HoJV13}
C.~Ho \emph{et~al.}, ``Adaptive task assignment for crowdsourced
  classification,'' in \emph{{ICML}}, 2013.

\bibitem{DBLP:journals/vldb/RoyLTAD15}
S.~B. Roy \emph{et~al.}, ``Task assignment optimization in knowledge-intensive
  crowdsourcing,'' \emph{{VLDB} J.}, 2015.

\bibitem{margaritis2000bayesian}
D.~Margaritis and S.~Thrun, ``Bayesian network induction via local
  neighborhoods,'' in \emph{Advances in neural information processing systems},
  2000.

\bibitem{DBLP:journals/corr/BiswasLR17}
S.~Biswas \emph{et~al.}, ``Combating the cold start user problem in model based
  collaborative filtering,'' \emph{CoRR}, vol. abs/1703.00397, 2017.

\bibitem{albert1972regression}
A.~Albert, \emph{Regression and the Moore-Penrose pseudoinverse}.\hskip 1em
  plus 0.5em minus 0.4em\relax Elsevier, 1972.

\bibitem{dismuke2006ordinary}
C.~Dismuke \emph{et~al.}, ``Ordinary least squares,'' \emph{Methods and Designs
  for Outcomes Research}, 2006.

\bibitem{fienberg1970iterative}
S.~E. Fienberg, ``An iterative procedure for estimation in contingency
  tables,'' \emph{The Annals of Mathematical Statistics}, 1970.

\bibitem{fienberg1983iterative}
S.~E. Fienberg and M.~M. Meyer, ``Iterative proportional fitting,''
  \emph{Encyclopedia of Statistical Sciences}, 1983.

\bibitem{mottin2013probabilistic}
D.~Mottin \emph{et~al.}, ``A probabilistic optimization framework for the
  empty-answer problem,'' \emph{Proceedings of the VLDB Endowment}, 2013.

\bibitem{garey1972optimal}
M.~R. Garey, ``Optimal binary identification procedures,'' \emph{SIAM Journal
  on Applied Mathematics}, vol.~23, no.~2, pp. 173--186, 1972.

\bibitem{pukelsheim2006optimal}
F.~Pukelsheim, \emph{Optimal design of experiments}.\hskip 1em plus 0.5em minus
  0.4em\relax SIAM, 2006.

\bibitem{de2007subset}
F.~De~Hoog and R.~Mattheij, ``Subset selection for matrices,'' \emph{Linear
  Algebra and its Applications}, 2007.

\bibitem{avron2013faster}
H.~Avron and C.~Boutsidis, ``Faster subset selection for matrices and
  applications,'' \emph{SIAM Journal on Matrix Analysis and Applications},
  2013.

\bibitem{niculescu2006bayesian}
R.~S. Niculescu \emph{et~al.}, ``Bayesian network learning with parameter
  constraints,'' \emph{Journal of Machine Learning Research}, 2006.

\bibitem{mead2010least}
J.~L. Mead and R.~A. Renaut, ``Least squares problems with inequality
  constraints as quadratic constraints,'' \emph{Linear Algebra and its
  Applications}, vol. 432, no.~8, pp. 1936--1949, 2010.

\bibitem{stark1995bounded}
P.~B. Stark and R.~L. Parker, ``Bounded-variable least-squares: an algorithm
  and applications,'' \emph{Computational Statistics}, 1995.

\bibitem{forman2003extensive}
G.~Forman, ``An extensive empirical study of feature selection metrics for text
  classification,'' \emph{Journal of machine learning research}, vol.~3, no.
  Mar, pp. 1289--1305, 2003.

\bibitem{DBLP:conf/sigmod/GuoPG12}
S.~Guo \emph{et~al.}, ``So who won?: dynamic max discovery with the crowd,'' in
  \emph{SIGMOD}, 2012.

\bibitem{DBLP:conf/webdb/PolychronopoulosADGP13}
V.~Polychronopoulos \emph{et~al.}, ``Human-powered top-k lists,'' in
  \emph{WebDB}, 2013, pp. 25--30.

\bibitem{DBLP:journals/tods/DavidsonKMR14}
S.~B. Davidson \emph{et~al.}, ``Top-k and clustering with noisy comparisons,''
  \emph{{ACM TODS}}, 2014.

\bibitem{DBLP:conf/pods/GrozM15}
B.~Groz and T.~Milo, ``Skyline queries with noisy comparisons,'' in
  \emph{{PODS}}, 2015, pp. 185--198.

\bibitem{DBLP:conf/edbt/2017}
H.~Rahman \emph{et~al.}, ``A probabilistic framework for estimating pairwise
  distances through crowdsourcing.'' in \emph{EDBT}, 2017.

\bibitem{DBLP:conf/cscw/ShawHC11}
A.~D. Shaw \emph{et~al.}, ``Designing incentives for inexpert human raters,''
  in \emph{{CSCW}}, 2011.

\bibitem{DBLP:conf/amcis/KaufmannSV11}
N.~Kaufmann \emph{et~al.}, ``More than fun and money. worker motivation in
  crowdsourcing - {A} study on mechanical turk,'' in \emph{AMCIS}, 2011.

\bibitem{GaoP14}
Y.~Gao \emph{et~al.}, ``Finish them!: Pricing algorithms for human
  computation,'' \emph{{PVLDB}}, 2014.

\bibitem{fan2015icrowd}
J.~Fan \emph{et~al.}, ``icrowd: An adaptive crowdsourcing framework,'' in
  \emph{SIGMOD}, 2015.

\bibitem{ece}
C.~H. Lin \emph{et~al.}, ``Signals in the silence: Models of implicit feedback
  in a recommendation system for crowdsourcing.'' in \emph{AAAI}, 2014.

\bibitem{us}
H.~Rahman \emph{et~al.}, ``Feature based task recommendation in crowdsourcing
  with implicit observations,'' \emph{HCOMP}, 2016.

\bibitem{kohonen1998self}
T.~Kohonen, ``The self-organizing map,'' \emph{Neurocomputing}, vol.~21, no.~1,
  pp. 1--6, 1998.

\end{thebibliography}


\end{document}